\documentclass[11pt]{article}

\usepackage{times}
\usepackage{chicagor}

\usepackage{amsmath, amsthm, amssymb}
\usepackage{mathrsfs}
 
\theoremstyle{plain}
\newtheorem{thm}{Theorem}[section]
\newtheorem{lem}[thm]{Lemma}
\newtheorem{cor}[thm]{Corollary}
\theoremstyle{definition}
\newtheorem{defn}[thm]{Definition}
\newtheorem{exmp}[thm]{Example}

\newcommand{\wbox}{\mbox{$\sqcap$\llap{$\sqcup$}}}

\newenvironment{RETHM}[2]{\trivlist \item[\hskip 10pt\hskip\labelsep{\sc #1\hskip 5pt\relax\ref{#2}.}]\it}{\endtrivlist}
\newcommand{\rethm}[1]{\begin{RETHM}{Theorem}{#1}}
\newcommand{\repro}[1]{\begin{RETHM}{Proposition}{#1}}
\newcommand{\relem}[1]{\begin{RETHM}{Lemma}{#1}}
\newcommand{\recor}[1]{\begin{RETHM}{Corollary}{#1}}
\newcommand{\erethm}{\end{RETHM}}
\newcommand{\erepro}{\end{RETHM}}
\newcommand{\erelem}{\end{RETHM}}
\newcommand{\erecor}{\end{RETHM}}

\newcommand{\IR}{\mbox{$I\!\!R$}}

\newcommand{\mm}{\mathit{mm}}

\newcommand{\BR}{\mathit{BR}}
\newcommand{\MSW}{\mathit{MSW}}
\newcommand{\coco}{\mathit{coco}}
\newcommand{\NE}{\mathit{NE}}
\newcommand{\BU}{\mathit{BU}}

\newcommand{\shortv}[1]{}
\newcommand{\fullv}[1]{#1}
\newcommand{\commentout}[1]{}

\newlength{\citeskipup}
\newlength{\citeskipdown}

\title{Cooperative Equilibrium: A Solution Predicting Cooperative
  Play
\thanks{Material in this paper appeared in preliminary form in two
  earlier papers: 
Cooperative equilibrium, 
{\em Proceedings of the Ninth International Joint Conference on
Autonomous Agents and Multiagent Systems (AAMAS 2010)}, 2010,
pp. 1465-1466, and Towards a deeper understanding of cooperative equilibrium:
characterization and complexity, {\em Proceedings of the Twelfth
International Joint Conference on Autonomous Agents and Multiagent
Systems (AAMAS 2013)}, 2013, pp. 319--326.}}
\author{Nan Rong  \ \ \ \ Joseph Y. Halpern\\
       Department of Computer Science\\
       Cornell University\\
       Ithaca, NY 14853\\
       \{rongnan,halpern\}@cs.cornell.edu}
\begin{document}
\maketitle
\begin{abstract}
Nash equilibrium (NE) assumes that players always make a best
response. However, this is not always true; sometimes people  
cooperate even it is not a best response to do
so. For example, in the Prisoner's  Dilemma, people often cooperate.
Are there rules underlying cooperative behavior? In an effort to
answer this question,  
we propose a new equilibrium concept: \textit{perfect cooperative
  equilibrium }(PCE), 
and two related variants: \emph{max-PCE} and \emph{cooperative equilibrium}.
PCE may help explain players' behavior in games where cooperation is
observed in practice. A player's payoff in a PCE is at least
as high as in any NE. However, a PCE does not always exist.  
We thus consider $\alpha$-PCE, where $\alpha$ takes into account the degree
of cooperation; a PCE is a 0-PCE.  Every game has a Pareto-optimal
\emph{max-PCE (M-PCE)}; 
that is, an $\alpha$-PCE for a maximum
$\alpha$.  We show that M-PCE does well at predicting behavior in quite
a few games of interest.  
We also consider \emph{cooperative equilibrium} (CE), 
another generalization of PCE that takes punishment 
into account.  Interestingly, all Pareto-optimal M-PCE are CE.  
We prove that, in 2-player games, both
a PCE (if it exists) and a M-PCE can be found in polynomial time,
using bilinear programming.  
This is a contrast to Nash equilibrium, which is PPAD complete even in
2-player games \cite{CDT09}.  
We compare M-PCE to the \emph{coco value} \cite{KK09}, another solution concept
that tries to capture cooperation, both axiomatically and in terms of an 
algebraic characterization,
and show that the two are closely related, despite their very different
definitions. 
\end{abstract}

\section{Introduction}
Nash Equilibrium (NE) assumes that players always make a best response to
what other players are doing. However, this assumption 
does not always hold.
Consider the Prisoner's Dilemma, in
which two prisoners can choose either to defect or to cooperate, with 
payoffs as shown in Table~\ref{table:PD}

\begin{table}[htb]
\begin{center}
\begin{tabular}{c |  c c}
& Cooperate & Defect\\
\hline
Cooperate &(3,3) &(0,5) \\
Defect  &(5,0)  &(1,1)  
\end{tabular}
\end{center}
\caption{Payoffs for Prisoner's Dilemma}\label{table:PD}
\end{table}

\noindent
Although the only best response here is to play Defect no matter what the
other player does, people often do play (Cooperate,
Cooperate). 

There are a number of other games in which Nash equilibrium
does not predict actual behavior well. 
To take one more example,
in the Traveler's Dilemma \cite{Basu94,Basu07}, 
two travelers have identical luggage, for which they 
paid the same price.  Their luggage is damaged (in an identical way)
by an airline.  The airline offers to recompense them for their luggage.
They may ask for any dollar amount between \$2 and \$100.
There is only
one catch.  If they ask for the same amount, then that is what they will
both receive.  However, if they ask for different amounts---say one asks
for \$$m$ and the other for \$$m'$, with $m < m'$---then whoever asks
for \$$m$ (the lower amount) will get \$$(m+2)$, 
while the other traveler will get \$$(m-2)$. 
A little calculation shows that the only NE in the Traveler's Dilemma is
$(2,2)$.  (Indeed, $(2,2)$ is the only strategy that survives iterated
deletion of weakly dominated strategies and is the only rationalizable
strategy; see \cite{OR94} for a discussion of these solution concepts.)
Nevertheless, in practice, people (even game theorists!) do not play
(2,2).  Indeed, when Becker, Carter, and Naeve \citeyear{BCN05} asked
members of the Game Theory Society to submit strategies for the game,
37 out of 51 people submitted a strategy of 90 or higher.
The strategy that was submitted most often (by 10 people) was 100.
The winning strategy (in pairwise matchups against all submitted
strategies) was 97. Only 3 of 51 people submitted the ``recommended''
strategy 2. In this case, NE is neither predictive nor 
normative; it is neither the behavior that was submitted most often
(it was in fact submitted quite rarely) nor the strategy that does best
(indeed, it did essentially the worst among all strategies submitted).

In both Prisoner's Dilemma and Traveler's Dilemma, people 
display what might be called ``cooperative'' behavior.
This cannot be explained by the best response assumption of NE. Are there rules
underlying cooperative behavior? 

In this paper, we propose a new solution concept, 
\textit{perfect cooperative equilibrium} (PCE), in an attempt to characterize
cooperative  behavior. 
Intuitively, in a 2-player game, a strategy profile (i.e., a strategy
for each player) is a PCE if each player does at least as well as she
would if the other player were best-responding. 
In Prisoner's Dilemma, both (Cooperate, Cooperate) and 
(Defect, Defect) are PCE.  To see why, suppose that the players are Amy and Bob. 
Consider the game from Amy's point of view.  She gets a
payoff of 3 from (Cooperate, Cooperate).  No matter what she
does, Bob's best response is Defect, which gives Amy a payoff
of either 0 or 1 (depending on whether she cooperates or defects).
Thus, her payoff with (Cooperate, Cooperate) is 
better than the payoff she would get with any strategy she could use, provided
that Bob best-responds.  The 
same is true for 
Bob.  Thus, (Cooperate, Cooperate) is a PCE.  The same argument
shows that (Defect, Defect) is also a PCE.  

This game already shows that some PCE are not NE.  In Traveler's
Dilemma, any strategy profile that gives each player
a payoff above 99 is a PCE (see Section \ref{sec:PCE}
for details). For example, both (99, 99) and (100, 100) 
are PCE. Moreover, the unique NE is not a PCE.  
Thus, in general, PCE and NE are quite different. 
We can in fact show that, if a PCE exists, the payoff for
each player is at least as good as it is in any NE.  This makes PCE an
attractive notion, especially for mechanism design.

This leads to some obvious questions.  First, \emph{why} should or do
players play (their part of) a PCE?  Second, does a PCE always exist?
Finally, how do players choose among multiple PCE, when more than one
exists?  

With regard to the first question, first consider one of the
intuitions for NE.  The assumption is that players have played
repeatedly, and thus have learned other players' strategies.  They thus
best respond to what they have learned.  A NE is a stable point of this
process: every players' strategy is already a best response to what the
other players are doing.   This intuition focuses on what players have
done in the past; with PCE, we also consider the \emph{future}.
In a PCE such as (Cooperate, Cooperate) in Prisoner's
Dilemma, players realize that if they deviate from the PCE, then the
other player may start to best respond;  after a while, they
may well end up in some NE, and thus have a payoff that is guaranteed
to be no better than (and is often worse than) that of the PCE.
Although cooperation here (and in 
other games) gives a solution concept that is arguably more ``fragile''
than NE, players may still want to play a PCE because it gives a better
payoff.  Of course, we are considering one-shot games, not repeated
games, so there is no future (or past); nevertheless, these intuitions
may help explain why players actually play a PCE.  
(See Section~\ref{sec:related} for a comparison of PCE and NE in
repeated games.)

It is easy to see that a PCE does not always exist.  Consider the
Nash bargaining game \cite{Nash50a}.  Each of two players 
requests a number of cents
between 0 and 100.  If their total request is no more than a dollar, then
they each get what they asked for; otherwise, they both get nothing.
Each pair $(x,y)$ with $x+y = 100$ is a NE, so
there is clearly no strategy profile that gives both players a higher
payoff than they get in every NE,
so a PCE does not exist.

We define a notion of \emph{$\alpha$-PCE}, where
$s$ is an $\alpha$-PCE if, playing $s$, each player can do at least $\alpha$
better than the best payoff she could get if the other player were
best-responding (note 
that $\alpha$ may be negative).  Thus,
if a strategy is an $\alpha$-PCE, then it is an $\alpha'$-PCE for all
$\alpha' \le \alpha$.  A strategy is a PCE iff it is a 0-PCE.  We are
most interested in \emph{max-perfect cooperative equilibrium (M-PCE)}.
A strategy is a M-PCE if it is an $\alpha$-PCE, and no strategy is an
$\alpha'$-PCE for some $\alpha' > \alpha$.  We show that every game has
a M-PCE; in fact, it has a \emph{Pareto-optimal} M-PCE (so that there
is no other strategy profile where all players do at least as well and
at least one does better).  
We show that M-PCE does well at predicting behavior
in quite a few games of interest.  For example,
in Prisoner's Dilemma, (Cooperate, Cooperate) is the unique M-PCE;
and in the Nash bargaining game,  $(50,50)$ is the unique
M-PCE.  As the latter example suggests, the notion of 
a M-PCE embodies a certain sense of fairness.  
In cases where there are several PCE, M-PCE gives a way of choosing
among them.

Further insight into M-PCE, at least in 2-player games, is provided by
considering another generalization of PCE, called
\emph{cooperative equilibrium} (CE), which takes
punishment into account.  It is well-known that people are willing to
punish non-cooperators, even at a cost to themselves (see, for example,
\cite{Hauert07,Sigmund07,QFTSSBF04} and the references therein).  
CE is defined only for 2-player games.
Intuitively, a strategy profile $s$ in a 2-player game is a CE if for each
player $i$ and each possible deviation $s_i'$ for $i$, either
(1) $i$ does at least as well with $s$ as she would do if the other
player $j$ were best-responding to $s_i'$; 
or (2) all of $j$'s best responses to $s_i'$ result in $j$
being worse off than he is with $s$, so he ``punishes'' $i$ by playing a
strategy $s_j''$ in response to $s_i'$ that results in $i$ being worse
off.  Note that it may be the case that by punishing $i$, $j$ is
himself worse off.  

It is almost immediate that every PCE is a CE. More interestingly, 
we show that every Pareto-optimal M-PCE is a CE. Thus, every 2-player
game has a CE. 
While CE does seem to capture reasoning often done by people, 
there are games where it does not have much predictive power.
For example, in the Nash bargaining game, CE and NE coincide;
all strategy profiles $(x,y)$ where $x+y = 100$ are CE.  
CE also has little predictive power in the 
\emph{Ultimatum} game \cite{GSS82},  a well-known variant of the
Nash bargaining game where player 1 moves first and proposes a division,
which player 2 can either accept or reject;
again, all offers give a CE.  In practice, ``unfair''
divisions (typically, where player 2 gets less than, say, 30\% of the
pot, although the notion of unfairness depends in part of cultural
norms) are rejected; player 2 punishes player 1 although he is worse

This type of punishment is not captured by CE, 
but can be understood in terms of M-PCE.
For example, a strategy in the
ultimatum game might be considered acceptable if it is close
to a M-PCE; that is, if a M-PCE is an $\alpha$-PCE, then a strategy
might be considered acceptable if it is an $\alpha'$-PCE, where $\alpha
- \alpha'$ is smaller than some (possibly culturally-determined) threshold.
Punishment is applied if the opponent's strategy precludes an acceptable
strategy being played.
To summarize, M-PCE is a solution concept that is well-founded, has good
predictive power, and may help explain when players are willing to apply
punishment in games.

Motivated by the attractive properties of PCE and M-PCE, we
analyze the complexity of finding a PCE or M-PCE. We prove that
in 2-player games, both 
a PCE and a M-PCE can be found in polynomial time, using bilinear
programming. We can also determine in polynomial time whether a PCE exists.
This is a contrast to Nash equilibrium, which is PPAD complete even in
2-player games \cite{CDT09}.  

We then compare M-PCE to 
other cooperative solutions. We focus on the \emph{coco
  (cooperative competitive)  value} \cite{KK09}, another solution
concept that tries to capture cooperative behavior in 2-player games. 
Because the coco value is not always achievable without side payments, 
in order to make a fair comparison, we consider games with side payments.
We provide a technique for converting a 2-player game without side
payments into one with side payments. 
We then compare M-PCE and the coco value both axiomatically and 
in terms of an algebraic characterization.
We show that, despite their quite different definitions, these two
notions are closely related.
They have quite similar algebraic characterizations involving maximum social
welfare and minimax values, and their axiomatic characterizations differ
in only one axiom.
The surprising similarities between M-PCE and coco value may lead to
insights for a deeper understanding of cooperative equilibrium in general.

\fullv{
The rest of the paper is organized as follows. 
In Section~\ref{sec:PCE}, we introduce PCE, prove its most important
properties, and give some examples to show how it works.
In Section~\ref{sec:A-PCE}, we consider $\alpha$-PCE and M-PCE;
in Section~\ref{sec:CE}, we consider CE.
We examine the complexity of finding a PCE/M-PCE/CE (and determining
whether a PCE exists) in 2-player  
games in Section~\ref{sec:complexity}.
In Section~\ref{sec:coco}, we compare M-PCE to the coco value.
We discuss relevant related work in Section~\ref{sec:related}.
}

\section{Perfect Cooperative Equilibrium}\label{sec:PCE} 
In this section, we introduce PCE.
For ease of exposition, we focus here on finite \emph{normal-form} games $G =
(N,A,u)$, where $N=\{1,\ldots,n\}$ is a finite set of players, $A = A_1
\times \ldots \times A_n$, $A_i$ is a finite set of possible actions for
player $i$, $u = (u_1, \ldots, u_n)$, and $u_i$ is player $i$'s utility
function, that is, $u_i(a_1, \ldots, a_n)$ is player $i$'s utility or payoff
if the action profile $a = (a_1, \ldots, a_n)$ is played. 
Players are allowed to randomize.  A strategy for player $i$ is thus a
distribution over actions in $A_i$; let $S_i$ represent the set of
player $i$'s strategies.  Let $U_i(s_1, \ldots, s_n)$ denote player
$i$'s expected utility if the strategy profile $s = (s_1, \ldots, s_n)$
is played.  Given a profile $x = (x_1, \ldots, x_n)$, let $x_{-i}$
denote the tuple consisting of all values $x_j$ for $j \ne i$.

\begin{defn} \label{d0} 
Given a game $G$, a strategy $\mathit{s_i}$ for player $i$ in $G$ is a
	\textit{best response} to a strategy $\mathit{s_{-i}}$ for the
	players in 
	$N-\left\{i\right\}$ if $s_i$ maximizes player $i$'s 
expected utility given that the other players are playing $s_{-i}$, that is,
	$U_i(s_i, s_{-i}) = \sup_{s_i' \in S_i} U_i(s_i', s_{-i})$.
Let $\BR^G_i(s_{-i})$ be the set of best responses to $s_{-i}$ in game
$G$.  We omit the superscript $G$ if the game is clear from context.
\end{defn}

We first define PCE for 2-player games.

\begin{defn}
Given a 2-player game $G$, let $\BU^G_i$ denote the best utility that
player $i$ 
can obtain if the other player $j$ best responds; that is,
$$\BU^G_i = \sup_{\{s_i \in S_i, s_{j} \in \BR^G(s_i)\}}  U_i(s).$$
(We again omit the superscript $G$ it if it is clear from context.)
\end{defn}

\begin{defn} \label{d1perfect}
A strategy profile $s$ is a \textit{perfect cooperative equilibrium (PCE)} 
in a 2-player game $G$ if, for all $i\in \{1,2\}$, we have
$$U_i(s) \ge \BU^G_i.$$
\end{defn} 

It is easy to show that every player does at least as well in a PCE as
in a NE.

\begin{thm} \label{prop1}
If $s$ is a PCE and $s^*$ is a NE in a 2-player game $G$, then for all
$i \in \{1,2\}$, we have $U_i(s) \ge U_i(s^*)$.
\end{thm}
\begin{proof} 
Suppose that $s$ is a PCE and $s^*$ is a NE.  Then, by the definition
of NE, $s^*_{3-i} \in BR(s_i^*)$, so by the definition of PCE, $U_i(s)
\ge U_i(s^*)$.
\end{proof}

It is immediate from Theorem~\ref{prop1} that a PCE does not always
exist.  For example, in the Nash bargaining game, a PCE would have to
give each player a payoff of 100, and there is no strategy profile that
has this property.  Nevertheless, we continue in this section to
investigate the properties of PCE; in the following two sections, we
consider generalizations of PCE that are guaranteed to exist.

A strategy profile $s$ \emph{Pareto dominates} strategy profile $s'$ if 
$U_i(s) \ge U_i(s')$ for all players $i$, strategy $s$ \emph{strongly
Pareto dominates} $s'$ if $s$ Pareto dominates $s'$ and $U_j(s) >
U_j(s')$ for some player $j$; strategy $s$ is \emph{Pareto-optimal} if
no strategy profile strongly Pareto dominates $s$; $s$ is a \emph{dominant
strategy profile} if it Pareto dominates all other strategy profiles.

A dominant strategy profile is easily seen to be a NE; it is also a PCE.
\begin{thm}\label{thm:dominant} If $s$ is a dominant strategy profile in a
2-player game $G$, then $s$ is a PCE.
\end{thm}
\begin{proof} 
Suppose that $s$ is a dominant strategy profile in $G$.  Then for all $i \in
\{1,2\}$, all $s'_i \in S_i$, and all $s'_{3-i} \in BR_{3-i}(s'_i)$, we have
that $U_i(s) \ge U_i(s')$.  Thus, $U_i(s)\ge BU_i$ for all
$i$, so $s$ is a PCE. 
\end{proof}

The next result shows that a strategy profile that Pareto dominates a PCE is
also a PCE.  Thus, if $s$ is a PCE, and $s'$ makes everyone at least as
well off, then  $s'$ is also a PCE.  Note that this property does not
hold for NE.  For example, in Prisoner's Dilemma, (Cooperate, Cooperate)
is not a NE, although it strongly Pareto dominates (Defect, Defect),
which is a NE.

\begin{thm} \label{prop1a}
In a 2-player game, a strategy profile that Pareto dominates a PCE
must itself be a PCE. 
\end{thm}
\begin{proof}
Suppose that $s$ is a PCE and $s^*$ Pareto dominates $s$.  
Thus,  for all $i \in N$, we have $$U_i(s^*) \ge U_i(s) \ge BU_i.$$
Thus, $s^*$ is a PCE.
\end{proof}

\begin{cor}\label{cor:pce} If there is a PCE in a 2-player game $G$, 
there is a Pareto-optimal PCE in $G$
(i.e., a PCE that is Pareto-optimal among all strategy profiles).
\end{cor}

\begin{proof} Given a PCE $s$, let $S^*$ be the set of strategy profiles
that Pareto dominate $s$.  This is a closed set, and hence compact.  Let
$f(s) = U_1(s) + U_2(s)$.  Clearly $f$ is a continuous function, so $f$
takes on its maximum in $S^*$; that is, there is some strategy $s^* \in
S^*$ such that $f(s^*) \ge f(s')$ for all $s' \in S^*$.
Clearly $s^*$ must
be Pareto-optimal, and since $s^*$ Pareto dominates $s$, it must be a
PCE, by Theorem~\ref{prop1a}.
\end{proof}

We now want to define PCE for $n$-player games, where $n > 2$.  The
problem is that ``best response'' is not well defined.  For example, in
a 3-player game, it is not clear what it would mean for players 2 and 3
to make a best response to a strategy of player 1, since what might be
best for player 2 might not be best for player 3.  We nevertheless want
to keep the intuition that player 1 considers, for each of her possible
strategies $s_1$, the likely outcome if she plays $s_1$.  If there is
only one other player, then it seems reasonable to expect that that player
will play a best response to $s_1$.  There are a number of ways we could
define an analogue if there are more than two players; we choose an
approach that both seems natural and leads to a straightforward
generalization of all our results.  Given an $n$-player game $G$ and a
strategy $s_i$ for player $i$, let $G_{s_i}$ be the $(n-1)$-player game
among the players in $N-\{i\}$ that results when player $i$ plays $s_i$.  
We assume that the players in $N-\{i\}$ respond to $s_i$ by playing some
NE in $G_{s_i}$.  Let $\NE^G(s_i)$ denote the NE of $G_{s_i}$.  Again,
we omit the superscript $G$ if it is clear from context.
We now extend the definition of PCE to $n$-player games for $n > 2$ by
replacing $\BR(s_i)$ by $\NE(s_i)$.
Note that 
if $|N| = 2$, then $\NE(s_i) = \BR(s_i)$, so this gives a generalization
of what we did in the 2-player case.
As a first step, we extend the definition of $\BU^G_i$ to the
multi-player case 
by using $\NE^G(s_i)$ instead of $\BR^G(s_i)$; that is,
$$\BU^G_i = \sup_{\{s \in S_i, s_{-i} \in \NE_i^G(s_i)\}}  U_i(s).$$

\begin{defn} \label{d2}
A strategy profile $s$ is a \textit{perfect cooperative equilibrium (PCE)} 
in a game $G$ if for all $i\in N$, we have $$U_i(s) \ge \BU^G_i.$$
\end{defn}

With this definition, we get immediate analogues of
Theorems~\ref{prop1}, \ref{thm:dominant},  
\ref{prop1a}, and Corollary~\ref{cor:pce}, with almost identical
proofs.  Therefore, we state the results here and omit the proofs.

\begin{thm} \label{prop2}
If $s$ is a PCE and $s^*$ is a NE in a game $G$, then for all
$i \in N$, we have $U_i(s) \ge U_i(s^*)$.
\end{thm}

\begin{thm}\label{thm:dominant1} If $s$ is a dominant strategy profile in a
game $G$, then $s$ is a PCE.
\end{thm}

\begin{thm} \label{prop2a}
A strategy profile that Pareto dominates a PCE
must itself be a PCE. 
\end{thm}

\begin{cor}\label{cor:pce1} If there is a PCE in a game 
$G$, there is a Pareto-optimal PCE in $G$.
\end{cor}

We now give some examples of PCE in games of interest.

\begin{exmp}
\textit{A coordination game:}
A coordination game has payoffs as shown in Table~\ref{table:coordination}.
\begin{table}[htb]
\begin{center}
\begin{tabular}{c | c c}
& $a$ & $b$\\
\hline
$a$ &$(k_1,k_2)$ & $(0,0)$ \\
$b$  &$(0,0)$  & $(1,1)$  
\end{tabular}
\caption{Payoffs for coordination game}\label{table:coordination}
\end{center}
\end{table}
It is well known that if $k_1$ and $k_2$ are both positive, then $(a,a)$
and $(b,b)$ are NE (there is also a NE  that  uses mixed
strategies).  On the other hand, if $k_1 > 1$ and $k_2 > 1$, then
$(a,a)$ is the only PCE; if $k_1 < 1$ and $k_2 < 1$, then $(b,b)$ is the
only PCE; and if $k_1 > 1$ and $k_2 < 1$, then there are no PCE (since,
by Theorem~\ref{prop1}, 
a PCE would simultaneously have to give player 1 a payoff of at least
$k_1$ and player 2 a payoff of at least 1).
\hfill \wbox
\end{exmp}

\begin{exmp}
\emph{Prisoner's Dilemma:}
Note that, in Prisoner's Dilemma, $\BU_1 = \BU_2 = 1$, since the best
response is always to defect.  Thus, a strategy profile $s$ is a PCE iff 
$\min(U_1(s),U_2(s)) \ge 1$.  It is immediate that 
(Cooperate, Cooperate) and (Defect, Defect) are PCE, and are the only
PCE in pure strategies, but there are other PCE in mixed strategies.  
For example, 
($\frac{1}{2}$Cooperate+$\frac{1}{2}$Defect, Cooperate) and  
($\frac{1}{2}$Cooperate+$\frac{1}{2}$Defect,
$\frac{1}{2}$Cooperate+$\frac{1}{2}$Defect) are PCE
(where $\alpha$Cooperate + $(1-\alpha)$Defect denotes the mixed strategy where
Cooperate is played with probability $\alpha$ and Defect is played with
probability $1-\alpha$).
\hfill \wbox
\end{exmp}

\begin{exmp}
\textit{Traveler's Dilemma: \ }
To compute the PCE for Traveler's Dilemma, we first need to compute
$\BU_1$ and $\BU_2$.  By symmetry, $\BU_1 = \BU_2$.  We now show that
$\BU_1$ is between $98\frac{1}{6}$ and 99.  If player 1 plays 
$\frac{1}{2}100 + \frac{1}{6}99 + \frac{1}{6}98 + \frac{1}{6}97$, then
it is easy to see that player 2's 
best responses are 99 and 98 (both give player 2 an expected payoff of
$98\frac{5}{6}$); player 1's expected payoff if player 2 plays 99 is
$98\frac{1}{6}$.  Thus, 
$\BU_1 \ge 98\frac{1}{6}$.  To see that $\BU_1$ is at most 99, suppose by
way of contradiction that it is greater than 99.  Then there must be
strategies $s_1 = p_{100}100 + p_{99}99 + \cdots +p_2 2 \in S_1$ and
$s_2 \in \BR_2(s_1)$ such that $U_1(s_1,s_2) > 99$.  
It cannot be the case that $s_2$ gives positive probability to 100 (for
then $s_2$ would not be a best response).  Suppose that $s_2$ gives positive
probability to 99.  Then 99 must itself be a best response.  Thus,
$U_2(s_1, 99) \ge U_2(s_1, 98)$, so $101p_{100} + 99p_{99} + 96p_{98}
\ge 100(p_{100} + p_{99}) + 98p_{98}$, so $p_{100} \ge p_{99} + 2p_{98}$.
Since a best response by player 2 cannot put positive weight on 100, the
highest utility that player 1 can get if player 2 plays a best response
is if player 2 plays 99; then $U_1(s_1, 99) \le 97p_{100} + 99p_{99} +
100p_{98} + 99(1-p_{100} - p_{99} - p_{98})$.  Since $U_1(s_1, 99) >
99$, it follows that $p_{98} > p_{100}$.  This gives a contradiction.
Thus, $s_2$ cannot give positive probability to 99.  This means
that $s_1$ does not give positive probability to either 100 or 99.  But
then $U_1(s_1,s_2) \le U_1(s_1, 98) \le 99$, a contradiction.

Since $s$ is a PCE if $U_i(s) \ge \BU_i(s)$, for $i = 1,2$, it follows
that the only PCE in pure strategies are $(100, 100)$ and $(99, 99)$.
There are also PCE in mixed strategies, such as 
$(\frac{1}{2}100+\frac{1}{2}99,\frac{1}{2}100+\frac{1}{2}99)$ and
$(100,\frac{2}{3}100 + \frac{1}{3}99)$.  
\hfill \wbox
\end{exmp}

\begin{exmp}
\textit{Centipede game: \ }
In the Centipede game \cite{rosenthal}, players take turns moving, with
player 1 moving at odd-numbered turns and player 2 moving at
even-numbered turns.  There is a known upper bound on the number of
turns, say 20.  At each turn $t < 20$, the player whose move it is can
either stop the game or continue.  At turn 20, the game ends if it has
not ended before then.  If the game ends after an odd-numbered turn $t$,
then the payoffs are $(2^t+1,2^{t-1})$; if the game ends after an
even-numbered turn $t$, then the payoffs are $(2^{t-1},2^t+1)$.  
Thus, if player 1 stops at round 1, player 1 gets 3 and player 2 gets 1;
if player 2 stops at round 4, then player 1 gets 8 and player 2 gets 17;
if player 1 stops at round 5, then player 1 gets 33 and player 2 gets
16.  If the game stops at round 20, both players get over 500,000.
The key point here is that it is always better for the player who moves
at step $t$ to end the game than it is to go on for one more step and
let the other player end the game.  Using this observation, a
straightforward backward induction shows the best response for a player
if he is called upon to move at step $t$ is to end the game.  
Not surprisingly,  the only Nash equilibrium has player 1 ending the
game right away.  
But, in practice, people continue the game for quite a while. 

We can think of the centipede game as a normal-form game, where
players are choosing strategies.
To compute the PCE for the game, we need to 
first compute $\BU_1$ and $\BU_2$.
If player 1 continues to the end of the
game, then player 2's best response is to also continue to the end of
the game, giving player 1 a payoff of $2^{19}$ (and player 2 a payoff of
$2^{20}+1$).  If we take $q_{i,j}$ to be the strategy where
player $i$ quits at turn $j$ and $q_{i,C}$ to be the strategy where player $i$
continues to the end of the game, then a straightforward computation
shows that $q_{2,C}$ continues to be a best response to $\alpha q_{1,19}
+ (1-\alpha)q_{1,C}$ as long as $\alpha \ge \frac{3\times
2^{18}}{3\times 2^{18}+1}$.  If we take $\alpha =  \frac{3\times
2^{18}}{3\times 2^{18}+1}$ and player 2 best responds by playing
$q_{2,C}$, then player 1's utility is $2^{19}+\frac{3\times
2^{18}}{3\times 2^{18}+1}$.  It is then straightforward to show that
this is in fact $\BU_1$.  A similar argument shows 
\fullv{that, if player 1 is
best responding, then the best player 2 can do is to play $\beta q_{2,18}
+ (1-\beta)q_{2,C}$, where $\beta = \frac{3\times 2^{17}}{3\times 2^{17}+1}$.
With this choice, player 1's best response is $q_{1,19}$.
using this strategy for player 2, we get that}
\shortv{that}
$\BU_2 =
2^{18}+\frac{3\times 2^{17}}{3\times 2^{17}+1}$.

It is easy to see that there is no pure strategy profile $s$ such that
$U_1(s) \ge \BU_1$ and $U_2(s) \ge \BU_2$.  However, there are many
mixed PCE.  
For example, every strategy profile $(q_{1,C}, s_2)$ where $s_2=\beta
q_{2,18}+(1-\beta)q_{2,C}$ and $\beta \in 
[1-\frac{3\times 2^{17}}{(3\times 2^{17}+1)(3\times 2^{18}+1)}, 
\frac{3\times 2^{18}}{3\times 2^{18}+1}]$ is a PCE. 
\hfill \wbox
\end{exmp}

While PCE has a number of attractive properties, and does seem to
capture some aspects of cooperative behavior, it does not always exist
In the next section, we consider
a variant of PCE that is guaranteed to exist.

\section{$\alpha$-Perfect Cooperative Equilibrium}\label{sec:A-PCE}
In this section, we start by considering a more quantitative version
of PCE called $\alpha$-PCE, which takes into account the degree of
cooperation exhibited by a strategy profile.

\begin{defn}\label{d4}
A strategy profile $s$ is an \emph{$\alpha$-PCE} in a game $G$ if $U_i(s) \ge
\alpha + \BU_i^G$ for all $i \in N$. 
\end{defn}

Clearly, if $s$ is an $\alpha$-PCE, then $s$ is an $\alpha'$-PCE for
$\alpha' \le \alpha$, and $s$ is a PCE iff $s$ is a 0-PCE.  
Note that an $\alpha$-PCE imposes some ``fairness'' requirements.  Each
player must get at least $\alpha$ more (where $\alpha$ can be negative)
than her best possible outcome if the other players best respond.

We again get analogues of Theorems~\ref{prop1} and 
\ref{prop1a}, and Corollary~\ref{cor:pce}, with similar proofs.

\begin{thm} \label{prop3}
If $s$ is an $\alpha$-PCE and $s^*$ is a NE in a game $G$, then for all
$i \in N$, we have $U_i(s) \ge \alpha +  U_i(s^*)$.
\end{thm}

\begin{thm} \label{prop3a}
A strategy profile that Pareto dominates an $\alpha$-PCE
must itself be an $\alpha$-PCE. 
\end{thm}

\begin{cor}\label{cor:pce3} If there is an $\alpha$-PCE in a game 
$G$, there is a Pareto-optimal $\alpha$-PCE in $G$.
\end{cor}

Of course, we are interested in $\alpha$-PCE with the maximum possible
value of $\alpha$.
\begin{defn}\label{d5}
The strategy profile $s$ is an \emph{maximum-PCE (M-PCE)} in a game $G$
if $s$ is an $\alpha$-PCE and for all $\alpha' > \alpha$, there is no
$\alpha'$-PCE in $G$.
\end{defn}

A priori, a M-PCE may not exist in a game $G$.  For example, it may be
the case that there is an $\alpha$-PCE for all $\alpha < 1$ without
there being a 1-PCE.  The next theorem, which uses the fact that the 
strategy space is compact, shows that this cannot be the case.

\begin{thm} \label{prop4}
Every game $G$ has a Pareto-optimal M-PCE.
\end{thm}
\begin{proof}
Let $f(s) = \min_{i\in N}(U_i(s)-\BU_i^G)$.
Clearly $f$ is a continuous function; moreover, if $f(s) = \alpha$, then
$s$ is an $\alpha$-PCE.  Since the domain consists of the set of
strategy profiles, which can be viewed as a closed subset of $[0,1]^{|A| 
\times N}$, the domain is compact.  Hence $f$ takes on its maximum at
some strategy profile $s^*$.  Then it is immediate from the definition that
$s^*$ is a M-PCE.
The argument that there is a Pareto-optimal M-PCE is essentially the
same as that given in Corollary~\ref{cor:pce} showing that there is a
Pareto-optimal PCE; we leave details to the reader.
\end{proof}

The following examples show that M-PCE gives some very reasonable
outcomes. 

\begin{exmp}
\textit{The Nash bargaining game, continued:} 
Clearly $U_1 = U_2 = 100$; $(50,50)$ is a $(-50)$-PCE and is the unique
M-PCE.  
\hfill \wbox
\end{exmp}

\begin{exmp}
\textit{A coordination game, continued:}
If $k_1 > 1$ and $k_2 > 1$, then $(a,a)$ is the unique M-PCE; 
if $k_1 < 1$ and $k_2 < 1$, then $(b,b)$ is the unique M-PCE.  In both
cases, $\alpha = 0$. If $k_1>1$ and $k_2<1$, 
then the M-PCE depends on the exact values of $k_1$ and $k_2$.  
If $k_1-1>1-k_2$, then $(a,a)$ is the unique M-PCE; if $k_1-1=1-k_2$,
then both $(a,a)$ and $(b,b)$ are 
M-PCE; otherwise, $(b,b)$ is the unique M-PCE. In all three cases,
$\alpha=-\min(k_1-1,1-k_2) <0$. 
\hfill \wbox
\end{exmp}

\begin{exmp}
\textit{Prisoner's Dilemma, continued: \ }
Clearly (Cooperate, Cooperate) is a 2-PCE and (Defect, Defect) is a
0-PCE; (Cooperate, Cooperate) is the unique M-PCE.
\hfill \wbox
\end{exmp}

\begin{exmp}
\textit{The Traveler's Dilemma, continued: \ }
$(100, 100)$ is easily seen to be the unique M-PCE; since there is no
strategy profile that guarantees both players greater than 100 (since
for any pair of pure strategies, the total payoff to the players is at
most 200, and the total payoff from a mixed strategy profile is a convex
combination of the payoff of pure strategy profiles).
\hfill \wbox
\end{exmp}

\begin{exmp}
\textit{The centipede game, continued: \ }
A straightforward computation shows that the M-PCE in this game is
unique, and is the strategy profile $s^*$ of the form 
$(\alpha q_{1,C} + (1-\alpha)q_{1,19}, q_{2,C})$, where $\alpha$ is
chosen so as to maximize $\min(U_1(s^*)-\BU_1,U_2(s^*)-\BU_2)$.  This
can be done by taking $\alpha=\frac{1}{3\times 2^{18}+2}-\frac{3\times
2^{17}}{(3\times 2^{18}+2)(3\times 2^{18}+1)(3\times 2^{17}+1)}$. 
\hfill \wbox
\end{exmp}

\section{Cooperative Equilibrium}\label{sec:CE}
We can gain further insight into M-PCE (and into what people actually do
in a game) by considering a notion that we call \emph{cooperative
equilibrium}, which generalizes PCE by allowing for the possibility of
punishment. We define CE for 2-player games.  
(As we discuss below, it is not clear
how to extend the definition to $n$-player games for $n > 2$.)

\begin{defn} \label{d1}
A strategy profile $s$ is a \textit{cooperative
 equilibrium (CE)} in a 2-player
game $G$ if, for all players $i \in \{1,2\}$ and
all strategies $s'_i \in S_i$, 
if $j$ is the player other than $i$,
one of the following conditions holds:
\begin{enumerate}
\item $U_i(s) \ge \sup_{s_j'\in BR_j(s_i')} U_i(s');$
\item $U_j(s)> \sup_{s_j'\in S_j}U_j(s')$, 
 		and for some $s'_j \in S_j$, we have
 $U_i(s)\geq U_i(s')$. 
\end{enumerate}
\end{defn}

If we consider only the first condition, then the definition would
be identical to PCE.  
It thus follows that all PCEs are CEs.
The second condition is where punishment comes in.  
Suppose that players $i$ and $j$ are Alice and Bob, respectively.
If there is no response that Bob can make to $s_i'$ that makes 
Bob better off than he is with $s$ then, intuitively, Bob becomes unhappy, 
and will seek to punish Alice. By ``punish Alice'', we mean that Bob will play
a strategy that  makes  Alice no better off than she is with $s$. 
We assume that if Bob can punish Alice when she plays $s_i'$,
then Alice will not deviate to $s_i'$.
In other words, $s$ is a CE if for all strategies $s_i'\in S_i$, 
Alice has no motivation to deviate to $s_i'$ either because (1) when Bob
best responds to $s_i'$, Alice is no better off than she is with $s$, or 
(2) Bob is strictly worse off even when he best responds
to $s_i'$, and Bob can punish Alice by playing a strategy which would make
Alice no better off than she is in $s$; and similarly with the roles
of Alice and Bob reversed.

We are not sure how to generalize CE to arbitrary games.  We could, of
course, replace $\BR_{j}(s'_i)$ by $\NE_{-i}(s'_i)$ in the first
clause.  The question is what to do in the second clause.  We could say
that if each player in $N-\{i\}$ is worse off in every Nash equilibrium
in the game $G_{s_i}$, they punish player $i$.  But punishment
may require a coordination of strategies, and it is not clear how the
players achieve such coordination, at least in a one-shot game.  Not
surprisingly, the examples in the literature where players punish others
are 2-player games like the Ultimatum game.  In general, the intuition of
punishment seems most compelling in 2-player games.

Our main interest in CE is motivated by the following result, which
shows that every Pareto-optimal M-PCE is a CE.

\begin{thm}\label{ParetoM-PCE} Every Pareto-optimal M-PCE is a CE.
\end{thm}

\begin{proof}
Suppose that $s$ is a Pareto-optimal M-PCE.  To see that $s$ is a CE,
consider the maximum $\alpha$ such that $s$ is an $\alpha$-PCE.  
If $\alpha \ge 0$, then $s$ is a PCE, and hence clearly a CE, so we are done.   
If $\alpha < 0$, then suppose by way of contradiction that $s$ is not a CE.
One of the players, say 1, must have a deviation to a strategy $s_1'$
such that either (1) player 2 has a best response $s_2'$ to $s_1'$ such that
$U_1(s') > U_1(s)$ and $U_2(s') \ge U_2(s)$ or (2) for all
$s_2' \in S_2$, we have $U_2(s') < U_2(s)$ and
$U_1(s')>U_1(s)$. Intuitively, case (2) says that player 2 does worse than $U_2(s)$ no matter
what he does, and cannot punish player 1.  In case (1), it is immediate that
$s$ is not a Pareto-optimal M-PCE.  So we need to consider only case (2).

Suppose that (2) holds.
By definition, $U_i(s) \ge \alpha + \BU_i$ for all $i\in  \{1, 2\}$.  
By compactness,
there must be a strategy profile $s^*$ such that $s_1^* \in
\BR_1(s_2^*)$ and  $U_2(s_2^*) = \BU_2$.  
We claim that $s^*$ is a $\beta$-PCE for some $\beta > \alpha$
(recall that $\alpha$ is the maximum $\alpha'$ such that $s$ is an
$\alpha'$-PCE), contradicting the assumption that $s$ is a M-PCE.   
Since $s_1^* \in
\BR_1(s_2^*)$, we must have $U_1(s^*) \ge U_1(s_1',s_2^*)$ (by the
definition of $\BR$); 
moreover, $U_1(s_1',s_2^*) > U_1(s)$ by case (2).
Since $U_1(s^*) \ge U_1(s_1',s_2^*)$ and $U_1(s_1',s_2^*) > U_1(s)$,
it follows that 
$U_1(s^*)>U_1(s)$. 
Since $U_1(s) \ge \alpha + \BU_1$, there must be some $\beta' > \alpha$
such that $U_1(s^*) \ge \beta' + \BU_1$.  By definition, $U_2(s^*) = \BU_2=0+\BU_2$.  
Thus, $s^*$ is a $\beta$-PCE, where $\beta=\min(\beta',0)$.  
Since $\alpha<0$ and $\alpha<\beta'$, 
we have that $\alpha<\min(\beta',0)=\beta$. Thus, the claim holds, 
completing the proof. 
\end{proof}

\fullv{
We can also prove the following analogues of Theorem~\ref{prop1a} and 
Corollary~\ref{cor:pce}.  Since the proofs are quite similar to 
proofs of Theorem~\ref{prop1a} and Corollary~\ref{cor:pce}, we omit them
here. 
\begin{thm} \label{prop:ce}
A strategy profile that Pareto dominates a CE must itself be a CE. 
\end{thm}

\begin{cor}\label{cor:ce} There is a Pareto-optimal CE in every game.
\end{cor}
}

We now consider how CE works in the examples considered earlier.

\begin{exmp}
\textit{The Nash bargaining game:} 
Recall that the Nash bargaining game does not have a PCE, and that every
profile of the form $(a, 100-a)$ is a NE.  We now show that 
each of these profiles is a CE as well.  To see this, first observe that
$U_1(s) + U_2(s) \le 100$ for any strategy profile $s$.  (This is
clearly true for pure strategy profiles, and the expected
utility of a mixed strategy profile is just the convex combination of the
utilities of the underlying pure strategy profiles.)  Now suppose that
player 1 deviates from $(a, 100-a)$ to 
some strategy $s_1$, and that player 2's expected utility from a best 
response $s_2'$ to $s_1$ is $b$.  If $b \ge 100 -a$, then 
$U_1(a,s_2') \le a$, and the first condition of CE applies.  If $b < 100-a$,
then player 2 can punish player 1 by playing 100, in which case player 1 always
gets a reward of 0, and the second condition of CE applies.  
The same considerations apply to player 2's deviations.  
Thus, $(a, 100-a)$ is a CE.  Only one of these CE is a M-PCE: $(50,50)$.

There are also Nash equilibria in mixed strategies; for example,
$(\frac{1}{3}25 + \frac{2}{3}75,\frac{1}{3}25 + \frac{2}{3}75)$ is a NE.
However, it is not hard to show that no nontrivial mixed strategy
profile (i.e., one that is not a pure strategy profile) is a CE.
For suppose that $s$ is a CE where either $s_1$ or
$s_2$ are nontrivial mixed strategies.  We show below that $U_1(s) +
U_2(s) < 100$. 
This means there is pair $(a, 100-a)$ such that $a >
U_1(s)$ and $100-a > U_2(s)$.  So if player 1 deviates to $a$ and
player 2 deviates to $100-a$, neither of the two conditions that
characterize CE hold.

It now remains to show that for nontrivial mixed strategy profiles $s$,
we have $U_1(s)+U_2(s)<100$. 
Suppose that $s_1$ is a nontrivial mixed strategy.
Let $s_1[a]$ denote the probability that $s_1$ plays the pure strategy $a$.
Then $U_1(s)=\sum_{\{a: s_1[a]>0\}}s_1[a]U_1(a,s_2)$, and 
$U_2(s)=\sum_{\{a: s_1[a]>0\}}s_1[a]U_2(a,s_2)$.
So $U_1(s)+U_2(s)=\sum_{\{a:s_1[a]>0\}}s_1[a](U_1(a,s_2)+U_2(a,s_2))$.
Recall that $U_1(s')+U_2(s')\le 100$ for all possible strategy profiles $s'$.
So $\sum_{\{a:s_1[a]>0\}}s_1[a](U_1(a,s_2)+U_2(a,s_2))\le 100$, with equality
holding only when $U_1(a,s_2)+U_2(a,s_2)=100$ for all $a$ such that $s_1[a]>0$.
By assumption, there are at least two strategies $a$ and $a'$ such
that $s_1[a] > 0$ and $s_1[a'] > 0$.  
As can be easily verified, we cannot have $U_1(a,s_2)+U_2(a,s_2)=
U_1(a',s_2)+U_2(a',s_2)=100$.
Thus $U_1(s)+U_2(s)<100$, as desired.
\hfill \wbox
\end{exmp}

\begin{exmp}
\textit{A coordination game, continued:}
If $k_1 > 1$ and $k_2 > 1$, then $(a,a)$ is the only CE; if 
$k_1 < 1$ and $k_2 < 1$, then $(b,b)$ is the only CE; if 
$k_1 > 1$ and $k_2 < 1$, then the two NE, $(a,a)$ and $(b,b)$, are both
CE (although neither is a PCE).  
There is one other NE $s$ in mixed strategies; $s$ is not a CE. 
To see this, note that in $s$ both players have to put positive
probability on each pure strategy.
It easily follows that $U_2(s) = U_2(s_1,b) < 1$ (since $s_1$ puts
positive probability on $a$); similarly, $U_1(s) < 1$.
Hence, if player 1 plays $b$ instead of $s_1$, player 2 has a
unique best response of $b$, 
which strictly increases both players' payoffs. 
Thus, $s$ is not a CE. 
\hfill \wbox
\end{exmp}

\begin{exmp}
\textit{Prisoner's Dilemma, continued:}
Clearly each PCE in Prisoner's Dilemma is a CE.
As we now show, no other strategy profile is a CE.  Suppose, by way of
contradiction, that $s$ is a CE that is not a PCE.  Then some player
must get a payoff with $s$ that is strictly less than 1.  Without loss
of generality, we can assume that it is player 1.
Suppose that $U_1(s) = r_1 < 1$.
But then if player 1 plays Defect, he is guaranteed a better payoff---at
least 1---no matter 
what player 2 does, so $s$ cannot be a CE.
\hfill \wbox
\end{exmp}

\begin{exmp}
\textit{The Traveler's Dilemma, continued:} 
Of course, every PCE in Traveler's Dilemma is a CE, but there are 
others.
For example, $(100, 99)$ is a CE but not a PCE. To see this, note that with
$(100, 99)$, player 1 gets a payoff of 97 and player 2 gets 101, the
maximum possible payoff. So player 2 has no motivation to deviate.
Suppose that there exists some strategy $s_1$ that gives player 1
a payoff strictly greater than 97 when player 2 best responds. This
strictly decreases player 2's payoff. However, player 2 can punish
player 1 by playing 2, so that player 1 gets at most 2, strictly less
than what he gets originally. It easily follows that (100, 99) is a CE.
A similar argument shows that every other Pareto-optimal strategy profiles is a 
CE. 

Recall that (100, 100) is the unique M-PCE of this game.  Intuitively,
a M-PCE satisfies fairness requirements that an arbitrary CE does not.
\hfill \wbox
\end{exmp}

\begin{exmp}
\textit{The centipede game, continued:}
Again, every PCE is a CE. 
In addition, every Pareto-optimal strategy profile is a CE.  Thus, for example,
the strategy profile where both players continue to the end of the game is a CE
(although it is not a PCE), as is the profile where 
player 2 continues at all his moves, but player 1 ends the game at his
last turn. To see that a Pareto-optimal strategy profile is a CE, let $s$ be a 
Pareto-optimal strategy profile.  By way of contradiction, suppose that $s$
is not a CE.  Then there must be a strategy $s_i'$ for some player $i$
such that either (1) there is a best response $s_{3-i}'$ to $s_i'$ 
such that 
$U_i(s) > U_i(s')$ and $U_{3-i}(s') \ge U_{3-i}(s)$ or (2) for all 
$s_{3-i}' \in S_{3-i}$, it must be the case that $U_{3-i}(s') <
U_{3-i}(s)$ and $U_i(s) < U_i(s')$; that is, player $3-i$ does worse than
$U_{3-i}(s)$ no matter 
what he does, and cannot punish player $i$.  In case (1), it is
immediate that $s$ is not Pareto optimal; and case (2) cannot hold,
since player $3-i$ can always punish player $i$ by exiting at his first
turn.  
\hfill \wbox
\end{exmp}

\section{The Complexity of Finding a PCE, M-PCE, and
  CE}\label{sec:complexity} 
In general, it is not obvious how a PCE (or M-PCE, or CE) can be found
efficiently. We show that in 2-player games, a PCE can be found in  
polynomial time if one exists; moreover, determining whether one exists
can also be done in polynomial time. 
Similarly, in 2-player games, both a M-PCE and a CE can always be found in polynomial time.  
The first step in the argument involves showing that in
2-player games, for all strategy profiles $s$, there is a strategy
profile $s'=(s'_1,s'_2)$ that Pareto dominates $s$ such that both
$s_1'$ and $s_2'$ have support at most two pure
strategies (i.e., they give positive probability to at most two pure
strategies).  
We then show that both the problem of computing a PCE and a M-PCE can be
reduced to solving a polynomial number of ``small'' bilinear programs,
each of which can be solved in constant time.  This gives us the desired
polynomial time algorithm for PCE and M-PCE. We then use similar techniques
to show that a Pareto-optimal M-PCE, and thus a CE,
can be found in polynomial time,

\noindent\textbf{Notation:} For a matrix $\mathbf{A}$, let
$\mathbf{A}^T$ denote $\mathbf{A}$ transpose, 
let $\mathbf{A}[i,\cdot]$ denote the $i$th row of $\mathbf{A}$, let
$\mathbf{A}[\cdot,j]$ denote the $j$th column of $\mathbf{A}$,  
and let $\mathbf{A}[i,j]$ be the entry in the $i$th row, $j$th column of
$\mathbf{A}$.  
We say that a vector $x$ is \emph{nonnegative}, denoted $x \ge 0$, if its
all of its entries are nonnegative.

We start by proving the first claim above.
In this discussion, it is convenient to identify a strategy for player 1
with a column vector in $\IR^n$, and a strategy for player 2 with a
column vector in $\IR^m$.  The strategy has a support of size at most two
if the vector has at most two nonzero entries.

\begin{lem}\label{LEM3}
In a 2-player game, for all strategy profiles $s^*$, there exists a
strategy profile $s'=(s_1',s_2')$ that Pareto dominates $s^*$ such that
both $s_1'$ and $s_2'$ have support of size at most two.
\end{lem}
See the appendix for the proof of this lemma and other results not
proved in the main text.

The rest of the section makes use of \emph{bilinear} programs.  
There are a number of slightly different definitions of ``bilinear program''.
For our purposes, we use the following definition.
\begin{defn}
A \emph{bilinear program} $P$ (of size $n \times m$) is a quadratic program
of the form 
$$\begin{array}{ll}
\text{maximize} &x^T\mathbf{A}y+ x^T  c + y^T c'\\
\text{subject to} &x^T \mathbf{B}_1 y \ge d_1\\ 
                  & \mathbf{B}_2 x = d_2\\ 
                  &\mathbf{B}_3 y = d_3\\
                  & x \ge 0\\                  
                  &y \ge 0,
\end{array}$$
where $\mathbf{A}$ and $\mathbf{B}_1$  are $n \times m$ matrices, $x, c
\in \IR^n$, $y, c' \in \IR^m$, $\mathbf{B}_2$ is a $k \times n$ matrix
for some $k$, and $\mathbf{B}_3$ is a $k' \times m$ matrix for some $k'$.
$P$ is \emph{simple} if $\mathbf{B}_2$ and $\mathbf{B}_3$ each has one
row, consisting of all 1's.
(Thus, in a simple bilinear program, we have a single bilinear constraint 
$x^T \mathbf{B}_1 y \ge d_1$, non-negativity constraints on $x$ and $y$,
and constraints on the sum of the components of the vectors $x$ and $y$;
that is, constraints of the form 
$\sum_{i=1}^n x[i] = d'$ and
$\sum_{j=1}^m y[j] = d''$.)
\hfill \wbox
\end{defn}

\begin{lem} \label{BP}
A simple bilinear program of size $2 \times 2$ can be solved in constant
time.  
\end{lem}

We can now give our algorithm for finding a PCE.
The idea is to first find $\BU_1$ and $\BU_2$, which can be done in
polynomial time.  We then
use Lemma \ref{LEM3} to reduce the problem to 
$(^n_2)(^m_2)=O(n^2m^2)$ smaller problems, each of a which is a simple
bilinear program of size $2 \times 2$.  By Lemma~\ref{BP}, each of these
smaller problems can be solved in constant time, giving us a
polynomial-time algorithm.
\begin{thm} \label{PCE}
Given a 2-player game $G = (\{1,2\}, A, u)$, we can compute in polynomial time
whether $G$ has a PCE and, if so, we can compute a PCE in 
polynomial time.
\end{thm}

The argument that a M-PCE can be found in polynomial time is very
similar. 

\begin{thm} \label{MPCE}
Given a 2-player game $G = (\{1,2\}, A, u)$, we can compute 
a M-PCE in polynomial time. 
\end{thm}

Again, we use similar arguments to show that a Pareto-optimal M-PCE, 
and thus CE, can be found in polynomial time.
\begin{thm} \label{POMPCE}
Given a 2-player game $G = (\{1,2\}, A, u)$, we can compute 
a Pareto-optimal M-PCE in polynomial time. 
\end{thm}

Since, by Theorem~\ref{ParetoM-PCE}, a Pareto-optimal M-PCE is a
(Pareto-optimal) CE, the 
following corollary is immediate.
\begin{cor} \label{CE}
Given a 2-player game $G = (\{1,2\}, A, u)$, we can compute 
a Pareto-optimal CE in polynomial time. 
\end{cor}

\section{Comparing M-PCE and the Coco Value}\label{sec:coco}
In this section, we compare M-PCE to the coco value,  
a solution concept proposed by Kalai and Kalai \citeyear{KK09}
that also tries to capture cooperation.
Since the coco value is only defined for 2-player games, we 
consider only 2-player games in this section.
We show that despite their definitions being quite different, the two
solution concepts are closely related.
We also consider their computational complexity, and show that both
can be solved in polynomial time in 2-player games.

\subsection{A review of the coco value}
The coco value is computed by decomposing a game into two components,
which can be 
viewed as a purely cooperative component and a purely competitive
component.  The cooperative component is a \emph{team game}, a game
where both players have identical utility matrices, so that both players
get identical payoffs, no matter what strategy profile is played.
The competitive component is a \emph{zero-sum} game, that is, one where
if player 1's payoff matrix is $A$, then player 2's payoff matrix is $-A$.

As Kalai and Kalai \citeyear{KK09} observe,
every game $G$ can be uniquely decomposed into a team game $G_t$ and a
zero-sum game $G_z$, where if $(\mathbf{A},\mathbf{B})$, $(\mathbf{C},\mathbf{C})$, and $(\mathbf{D},-\mathbf{D})$ are the
utility matrices for $G$, $G_t$, and $G_z$, respectively, then $\mathbf{A}= \mathbf{C}+\mathbf{D}$
and $\mathbf{B}=\mathbf{C}-\mathbf{D}$.
Indeed, we can take $\mathbf{C}= (\mathbf{A} + \mathbf{B})/2$ and 
$\mathbf{D} = (\mathbf{A} - \mathbf{B})/2$.
We call $G_t$ \emph{the team game of $G$} and call $G_z$ \emph{the zero-sum
game of $G$}. 

The \emph{minimax value of game $G$ for player $i$}, denoted 
${\mm}_i(G)$, is the payoff player $i$ gets when the opponent is minimizing 
$i$'s maximum payoff; formally, 
$${\mm}_1(G)= \min_{s_{2}\in S_{2}} \max_{s_1 \in S_1} U_1(s_1,s_2);$$
$\mm_2(G)$ is defined similarly, interchanging 1 and 2.

We are now ready to define the coco value.  Given a game
$G$, let $a$ be the
largest value obtainable in the team game $G_t$ (i.e., the largest value in
the utility matrix for $G_t$), and let $z$ be the minimax value for
player 1 in the zero-sum game $G_z$.  Then 
the \emph{coco value of $G$}, denoted $\coco(G)$, is 
$$(a+z,a-z).$$ 
Note that the coco value is attainable if utilities are transferable:
the players simply play the strategy profile that gives the value $a$ in
$G_t$; then player 2 transfers $z$ to player 1 ($z$ may be negative, so
that 1 is actually transferring money to 2).  Clearly this outcome 
maximizes social welfare.  Kalai and Kalai \citeyear{KK09} argue that it
is also fair in an appropriate sense.  

\subsection{Examples}
The coco value and M-PCE value are closely related in a number of games of interest, as the following
examples show.
\begin{exmp}
\textit{The Nash bargaining game, continued:}
Clearly, the largest payoff obtainable in 
the team game corresponding to the Nash Bargaining game is $(50,50)$.
Since the game is symmetric, 
the minimax value of each player in the zero-sum game is 0.
Thus, the coco value 
of the Nash bargaining game
is $(50,50)$, which, as we have seen, is also the unique M-PCE value.
\hfill \wbox
\end{exmp}

\begin{exmp}
\textit{Prisoner's Dilemma, continued: \ }
Clearly, the largest payoff obtainable in the team game 
corresponding to Prisoner's Dilemma (given the payoffs shown in the
Introduction) 
is $(3,3)$. Since the game is symmetric,
again, the minimax value in the corresponding zero-sum game is 0.
Thus, the coco value is $(3,3)$, which is 
easily seen to also be the unique M-PCE value:
with these payoffs, $\BU_1 = \BU_2 = 1$, so by both cooperating, the
players have a 2-PCE, which is clearly also a M-PCE.
\hfill \wbox
\end{exmp}

\begin{exmp}
\textit{Traveler's Dilemma, continued:\ }
Clearly, the largest payoff obtainable in 
the team game corresponding to the Traveler's Dilemma is $(100,100)$.
And again, since the game is symmetric, 
the minimax value for each player in the zero-sum game is 0.
Thus, the coco value is $(100,100)$, which 
is  also the unique M-PCE value.
\hfill \wbox
\end{exmp}

As the next example shows, there are games in which the coco value and
\begin{exmp}\label{xam:centipede1}
\textit{The centipede game, continued:\ }
It is easy to see that the largest payoff obtainable in 
the team game corresponding to the centipede game
is $(\frac{2^{19}+2^{20}+1}{2},$ $\frac{2^{19}+2^{20}+1}{2})$:
both players play to the end of the game and split the total payoff.
It is also easy to compute that, in the zero-sum game corresponding to
the centipede game,
player 1's minimax value is 1, while player 2's minimax
value is $-1$, obtained when both players quit immediately.
Thus, the coco value is $(\frac{2^{19}+2^{20}+1}{2}+1,\frac{2^{19}+2^{20}+1}{2}-1)$
$=(\frac{2^{19}+2^{20}+3}{2},\frac{2^{19}+2^{20}-1}{2})$.
This value is not achievable without side payments, and is higher than  
the M-PCE value. 
\hfill \wbox
\end{exmp}

Although, as Example~\ref{xam:centipede1} shows, the M-PCE value and
the coco value can differ, we can say more.
Part of the problem in the centipede game is that the computation
of the coco value effectively assumes that side payments are possible.
The M-PCE value does not take into account the possibility of side payments.
Once we extend the centipede game to allow side payments in an
appropriate sense, it turns out that the coco value and the M-PCE value
are the same.  
To do a fairer comparison of the M-PCE and coco values, we consider
games with side payments.

\subsection{2-player games with side payments}\label{side}
In this subsection, we describe how an arbitrary 2-player game without
payments can be transformed into a game with side payments.  
There is more than one way of doing this; we focus on one, and briefly
discuss a second alternative.  Our procedure
may be of interest beyond the specific application to coco and M-PCE.
We implicitly assume throughout that outcomes can be expressed in dollars
and that players value the dollars the same way.
The idea is to add strategies to the game that allow players to propose
``deals'', which amount to a description of what strategy profiles
should be played and how much money should be transferred.  If the
players propose the same deal, then the suggested strategy profile is
played, and the money is transferred.  Otherwise, a ``backup'' action
is played.

Given a 2-player game $G=(\{1,2\},A,u)$, let $G^* =
(\{1,2\},A^*,u^*)$ be the \emph{game with side payments extending $G$},
where $A^*$ and $u^*$ are defined as follows.  $A^*$ extends $A$ by
adding a collection of actions that we call \emph{deal actions}.
A deal action for player $i$ is a triple of the form $(a,r,a_i') \in
A \times \IR \times A_i$.  Intuitively, this action proposes that the
players play the action profile $a$ and that player 1 should transfer
$r$ to player 2; if the deal is not accepted, then player $i$ plays
$a_i'$.    
Given this intuition, it should be clear how $u^*$ extends $u$.  For
action profiles  $a \in A$, $u^*(a) = u(a)$.  
For profiles actions $a \in (A^*_1 - A_1) \times (A^*_2 - A_2)$, 
the players agree on a
deal if they both propose a deal strategy with the same first two
components $(a,r)$.  In this case they play $a$ and $r$ is transferred.
Otherwise, players just play the backup action. More precisely,
for $a,a'\in A$, $b_i\in A_i$, and $r,r'\in \IR$:
\begin{itemize}
\item $u^*(a) = u(a)$;
\item $u^*_1((a,r,b_1), (a,r,b_2))= u_1(a) -r$;\\
$u^*_2((a,r,b_1), (a,r,b_2))= u_2(a) +r$;  
\item $u^*((a,r,b_1),(a',r',b_2)) = u(b_1,b_2)$ if $(a,r) \ne (a',r')$;
\item $u^*((a,r,b_1),b_2) = u^*(b_1,(a',r',b_2)) = u(b_1,b_2)$.
\end{itemize}
As usual, players are allowed to randomize, and a strategy of player $i$
in $G^*$ is a distribution over actions in $A^*_i$; let $S^*_i$ represent 
the set of player $i$'s strategies.
Let $U_i^*(s)$ denote player $i$'s expected utility if the strategy profile
$s\in S^*$ is played.
We call $G^*$ the game with side payments \emph{extending} $G$, and call $G$
the game \emph{underlying} $G^*$.

Intuitively, when both players play deal actions, we can think of them 
as giving their actions to a
trusted third party.  If they both propose the same deal, the
third party ensures that the deal action is carried out and the transfer
is made. 
Otherwise, the appropriate backup actions are played.

In our approach, we have allowed players to propose arbitrary backup
actions in case their deal offers are not accepted.  We also
considered an alternative approach, where if a deal is proposed by one
of the parties but not accepted, then the players get a fixed default
payoff (e.g., they could both get 0, or a default strategy could be
played, and the players get their payoff according to the default
strategy). Essentially the same results as those we prove hold for this 
approach as well; see the end of Section~\ref{cfr}.

\subsection{Characterizing the coco value and the M-PCE value algebraically}\label{cfr}
At first glance, the coco value and the M-PCE value seem quite
different, although both are trying to get at the notion of cooperation.
However, we show below that both have quite similar characterizations.
In this section, we characterize the two
notions algebraically, using two similar formulas involving the maximum
social welfare 
and the minimax value.  In the next section, we compare axiomatic
characterizations of the notions.

Before proving our results, we first show that, although they are
different games, $G$ and $G^*$ agree on the relevant parameters
(recall that $G^*$ is the game with side payments extending $G$).
Let $\MSW(G)$ be the maximum social welfare of $G$; formally,
$\MSW(G) =\max_{a\in A} (u_1(a)+u_2(a))$.  

\begin{lem}\label{MINMAX} For all 2-player games $G$, 
$\MSW(G) = \MSW(G^*)$ and $\mm_i(G^*) =
\mm_i(G)$, for $i = 1,2$.  
\end{lem}
\begin{proof}

To see that $\MSW(G) = \MSW(G^*)$, observe
that, by the definition of $u^*$, 
for all action profiles $a^* \in A^*$, there exists an action profile
$a\in A$  
and $r \in \IR$ such that $u^*(a^*) = (u_1(a) + r, u_2(a) - r)$, 
so $u^*_1(a^*) + u_2^*(a^*) = u_1(a) + u_2(a)$.  

To see that $\mm_1(G^*) = \mm_1(G)$, observe that for 
all $t \in S_2$, $a\in A$, and $a_1'\in A_1$, we have that
$U_1^*((a,r,a_1'),t) = U_1(a_1',t)$
so $$\max_{a_1' \in A_1^*} U_1^*(a_1',t) = \max_{a_1' \in A_1} U_1(a_1',t).$$
It then follows that 
$$\max_{s_1 \in S_1^*} U_1^*(s_1,t) = \max_{s_1 \in S_1} U_1(s_1,t).$$
Thus,
$$\min_{t \in S_2} \max_{s_1 \in S_1^*} U_1^*(s_1,t) = 
	\min_{t \in S_2}\max_{s_1 \in S_1} U_1(s_1,t).$$
Therefore, 
$$\begin{array}{ll}
\mm_1(G^*) &= \min_{t \in S_2^*} \max_{s_1 \in S_1^*} U_1^*(s_1,t) \\
&\le \min_{t \in S_2}\max_{s_1 \in S_1^*} U_1^*(s_1,t) \ \ \ \ [\text{since }S_2^*\supset S_2]\\
&= \min_{t \in S_2} \max_{s_1 \in S_1} U_1(s_1,t)\\
&= \mm_1(G).
\end{array}$$
Thus, $\mm_1(G^*) \le \mm_1(G)$.   Similarly,
for all $s_1 \in S_1$, we have 
$\min_{a_2 \in A_2^*} U_1^*(s_1 ,a_2) = \min_{a_2 \in A_2} U_1(s_1 ,a_2)$.
It then follows that 
$\min_{t \in S_2^*} U_1^*(s_1 ,t) = 
\min_{t \in S_2} U_1(s_1 ,t)$.  Thus, 
$$\min_{t \in S_2^*} \max_{s_1\in S_1} U_1^*(s_1 ,t) = 
\min_{t \in S_2} \max_{s_1 \in S_1}U_1(s_1,t).$$  
It follows that 
$$\begin{array}{lll}
\mm_1(G^*) &= &\min_{t \in S_2^*} \max_{s_1
\in S_1^*} U_1^*(s_1 ,t) \\
&\ge &\min_{t \in S_2^*} \max_{s_1
\in S_1} U_1^*(s_1,t) \ \ \ \ [\text{since }S_1^*\supset S_1]\\
&= &\min_{t \in S_2} \max_{s_1 \in S_1} 
U_1(s_1,t)\\
&= &\mm_1(G).
\end{array}
$$
Thus, $\mm_1(G^*) = \mm_1(G)$.
A similar argument shows that $\mm_2(G^*) = \mm_2(G)$.
\hfill \wbox
\end{proof}

We now characterize the coco value.
\begin{thm}\label{coco}
If $G$ is a 2-player game, then $\coco(G) = \left(\frac{\MSW(G) +
\mm_1(G_z) -  \mm_2(G_z)}{2}, \frac{\MSW(G) - \mm_1(G_z)+
\mm_2(G_z)}{2}\right)$.%
\footnote{Note that $\mm_1(G_z) = -\mm_2(G_z)$ by
von Neumann's minimax theorem \cite{neumann28} (which says that
in every 2-player zero-sum games, there is an equilibrium where 
both players play a minimax strategy).
We write the expression in the form above
to better compare it to the M-PCE value.}
Moreover, $\coco(G) = \coco(G^*)$.
\end{thm}
\begin{proof}
It is easy to see that the Pareto-optimal payoff profile in $G_t$ is
$\left(\frac{\MSW(G)}{2},\frac{\MSW(G)}{2}\right)$. Thus, by definition, 
$$\begin{array}{ll}
&\coco(G)\\
=&\left(\frac{\MSW(G)}{2},\frac{\MSW(G)}{2}\right)+\left(\mm_1(G_z),\mm_2(G_z)\right)\\  
=&\left(\frac{\MSW(G)+2\mm_1(G_z)}{2},\frac{\MSW(G)+2\mm_2(G_z)}{2}\right)\\ 
=&\left(\frac{\MSW(G)+\mm_1(G_z)-\mm_2(G_z)}{2},\frac{\MSW(G)-\mm_1(G_z)+\mm_2(G_z}{2}\right)  
\end{array}$$
The last equation follows since $G_z$ is a zero-sum game, so
$\mm_1(G_z) =-\mm_2(G_z)$.

The fact that $\coco(G) = \coco(G^*)$ follows from the characterization
of $\coco(G)$ above, the fact that $\MSW(G) = \MSW(G^*)$
(Lemma~\ref{MINMAX}), and the fact that $(G_z)^* = (G^*)_z$, which we
leave to the reader to check.
\hfill \wbox
\end{proof}

The next theorem provides an analogous characterization of the M-PCE
value in 2-player games with side payments. It shows that in
such games the M-PCE value is unique and has the same form as the coco
value. Indeed, the only difference is that we replace $\mm_i(G_z)$
by $\mm_i(G)$.

\begin{thm}\label{t3} If $G$ is a 2-player game, then the unique M-PCE
value of 
the game $G^*$ with side payments extending $G$ is 
$\left(\frac{\MSW(G) + \mm_1(G) - \mm_2(G)}{2}, 
\frac{\MSW(G) - \mm_1(G)+ \mm_2(G)}{2}\right)$.
\end{thm}
\begin{proof}

We first show that $\BU^{G^*}_1 = \MSW(G) - \mm_2(G)$ and $\BU^{G^*}_2 =
\MSW(G) - \mm_1(G)$.  For $\BU^{G^*}_1$, let $a^*$ be an action profile
in $G$ that maximizes social welfare, that is, $U_1(a^*) + U_2(a^*) =
\MSW(G)$, and let $(s_1',s_2')$ be a strategy profile in $G$ 
such that $s_2' \in \BR^G(s_1')$ and
$U_2(s_1',s_2') = \mm_2(G)$.  (Thus, by playing $s_1'$, player 1 ensures
that player 2 can get no more utility than $\mm_2(G)$, and by playing
$s_2'$, player 2 ensures that she does get utility $\mm_2(G)$ when
player 1 plays $s_1'$.)  

Let $s=(s_1,s_2)$ be such that, in $s_1$, 
player 1 plays deal action $(a^*,\mm_2(G)-u_2(a^*),a_1')$
with the same probability that she plays $a_1'$ in $s_1'$
(where $s_1'$ is as defined above) for all $a_1'\in A_1$; and
$s_2=(a^*,\mm_2(G)-u_2(a^*),a_2)$ for some fixed $a_2\in A_2$.  
Intuitively, $s_1$ does the following: if player 2 agrees to the 
deal in $s_1$, then $a^*$ is carried out,
and player 1 transfers $\mm_2(G)-u_2(a^*)$ to player 2; otherwise player 1
plays the mixed strategy $s_1'$. $s_2$ is a deal action 
that agrees to $s_1$. 
Thus, $U_1^*(s) = u_1(a^*) - (\mm_2(G) - u_2(a^*)) = U_1(a^*) + u_2(a^*)
- \mm_2(G) = 
\MSW(G) - \mm_2(G)$, and $U_2^*(s) = \mm_2(G)$.   
On the other hand, if player 2 plays an action $a_2 \in A_2$, then
$U_2^*(s_1,s_2) = U_2(s_1',a_2)\le U_2(s')=\mm_2(G).$
Thus, player 2 gets at most $\mm_2(G)$ when player 1 plays $s_1$, 
so $s_2 \in \BR^{G^*}_2(s_1)$.  
This shows that $\BU_1^{G^*} \ge \MSW(G) - \mm_2(G)$.

To see that $\BU_1^{G^*} \le \MSW(G) - \mm_2(G)$, consider a strategy
profile $s'' = (s_1'',s_2'') \in S^*$ with $s_2'' \in
\BR^{G^*}_2(s_1'')$.  Since $\mm_2(G^*) = \mm_2(G)$, it follows that 
$U_2^*(s'') \ge \mm_2(G)$.  Since $\MSW(G^*) = \MSW(G)$ by
Lemma~\ref{MINMAX}, it follows 
that $U_1^*(s'') + U_2^*(s'') \le \MSW(G)$.  Thus, $U_1^*(s'') \le \MSW(G)
- \mm_2(G)$, so $\BU_1^{G^*} \le \MSW(G) - \mm_2(G)$.
Thus, $\BU_1^{G^*} = \MSW(G) - \mm_2(G)$, as desired.

The argument that $\BU_2^{G^*} = \MSW(G) - \mm_1(G)$ is similar.

Now suppose that we have a strategy $s^+ \in S^*$ such that $U_1(s^+)
\ge \BU_1^{G^*} + \alpha$ and $U_2^*(s^+) \ge \BU_2^{G^*} + \alpha$.
Since $\MSW(G^*) = \MSW(G)$, it follows that $\BU_1(G^*) + \BU_2(G^*) +
2\alpha \le \MSW(G)$.  Plugging in our characterizations of $\BU_1(G^*)$
and $\BU_2(G^*)$, we get that $\alpha \le \frac{-\MSW(G) + \mm_1(G) +
\mm_2(G)}{2}$.  Taking $\beta= \frac{-\MSW(G) + \mm_1(G) +
\mm_2(G)}{2}$, we now show that we can find a $\beta$-PCE.  It follows
that this must be a M-PCE.

Let $a^*$  be the action profile in $G$ defined above that
maximizes social welfare, and let $a'\in A$.  
Let $s^+=(s_1^+,s_2^+)$, where $s_1^+=(a^*,
u_1(a^*)-\frac{\MSW(G)+\mm_1(G)-\mm_2(G)}{2}, a_1')$ and
$s_2^+=(a^*, u_1(a^*)-\frac{\MSW(G)+\mm_1(G)-\mm_2(G)}{2}, a_2')$.
It is also easy to check that $U_1(s^+)=\frac{\MSW(G)+\mm_1(G)-\mm_2(G)}{2}$, 
and $U_2(s^+) = \frac{\MSW(G)-\mm_1(G)+\mm_2(G)}{2}$. 

It can also easily be checked that $U_i(s^+)=\BU_i+\beta$ for $i=1,2$, so
$s^+$ is indeed a $\beta$-PCE.  
Therefore, $s^+$ is a M-PCE, and its value is a M-PCE value, as desired. 
Since $U_1(s^+)+U_2(s^+)=MSW(G)$, it follows that the M-PCE value is unique.
\end{proof}

As Theorems~\ref{coco} and~\ref{t3} show, in a 2-player game $G^*$ with
side payments, the coco value and M-PCE value are characterized by very
similar equations, making use of $\MSW(G^*)$ and minimax values.
The only difference is that the coco value uses the minimax value of the
zero-sum game $G_z$,   
while the M-PCE value uses minimax value of $G$. 
It immediately follows from Theorem \ref{coco} and \ref{t3} that the coco
value and the M-PCE value coincide in all games where  
$$\mm_1(G_z)-\mm_2(G_z)=\mm_1(G)-\mm_2(G).$$
Such games include team games, \emph{equal-sum games} 
(games with a payoff matrices $(A,B)$ such that 
$A+B$ is a constant matrix, all of whose entries are identical), 
\emph{symmetric games} 
(games where the strategy space is the same for both players, that is,
$S_1 = S_2$, and $U_1(s_1,s_2)=U_2(s_2,s_1)$ for all $s_1,s_2\in S_1$),
and many others.  
We can also use these theorems to
show that the M-PCE value and the coco value can differ, even in a game
where side payments are allowed,
as the following example shows.
\begin{exmp}\label{xam1}
Let $G$ be the 2-player game described by the following payoff matrix:
\begin{table}[htb]
\begin{center}
\begin{tabular}{c |  c c}
& a & b\\
\hline
c &(3,2) &(1,0) 
\end{tabular}
\end{center}
\end{table}
\vspace{-.1in}

\noindent Let $G^*$ be the game with side payments extending $G$.
Taking player 1 to be the row player and player 2 to be the column player, 
it is easy to check that $\MSW(G) = 5$, $\mm_1(G) = 1$, and
%
$\mm_2(G)=2$,  
Thus, by Theorem \ref{MPCE}, the M-PCE value of $G^*$ is 
$(\frac{5+1-2}{2},\frac{5-1+2}{2})=(2,3)$.  
On the other hand, it is easy to check that $\coco(G) = \coco(G^*)= (3,2)$.  

It seems somewhat surprising that the M-PCE here should be $(2,3)$,
since player 1 gets a higher payoff than player 2 no matter which
strategy profile in $G$ is played.  Moreover, $\BU_1^G = 3$ and $\BU_2^G
= 2$.  But things change when transfers are allowed.  It is easy to
check that it is still the case that $\BU_1^{G^*}= 3$; if player 1 plays
$c$, then player 2's best response is to play $a$.  
But $\BU_2^{G^*} = 4$; if player 2 plays
$((c,a),2,b)$, offering to play $(c,a)$, provided that player 1
transfers an additional 2, then player 1's best response is to agree (for
otherwise player 2 plays $b$), giving player 2 a payoff of 4. 
The possibility that player 2 can ``threaten'' player 1 in this way (even
though the moves are made simultaneously, so no actual threat is
involved) is why $\mm_2(G) \ge \mm_1(G)$.
\hfill  \wbox
\end{exmp}

We conclude this subsection by considering what happens if a default
strategy profile is used instead of backup actions when defining
games with side payments.  
Let the default payoffs be $(d_1,d_2)$.
Then a similar argument to above shows that the M-PCE value becomes
$$\left(\frac{\MSW(G)+d_1 - d_2}{2}, \frac{\MSW(G)-d_1+d_2}{2}\right).$$ 
Thus, rather than using the minimax payoffs in the formula, we now use 
the default payoffs.  Note that if the default payoffs are $(0,0)$,
then the M-PCE amounts to the players splitting the maximum social
welfare. 
We leave the details to the reader.

\subsection{Axiomatic comparison}
In this section, we provide an axiomatization of the M-PCE value and
compare it to the axiomatization of the coco value given by Kalai and
Kalai \citeyear{KK09}.
Before jumping into the axioms, we first explain the term
``axiomatize'' in this context.
Given a function $f:A\rightarrow B$, we say a set AX of axioms 
\emph{axiomatizes $f$ in $A$} if $f$ is the unique function mapping $A$
to $B$ that satisfies all axioms in AX.
Recall that every 2-player normal-form game has a unique coco
value. We can thus view the coco value as a function from 
2-player normal-form games to $\IR^2$. Therefore, a set AX of axioms 
axiomatizes the coco value if the coco value is the unique
function that maps from the set to $\IR^2$ that satisfies all the axioms in
AX.

Kalai and Kalai \citeyear{KK09} show that the following collection of
axioms axiomatizes the coco value.  We describe the axioms in terms of
an arbitrary function $f$.  If $f(G) = (a_1, a_2)$, then we take $f_i(G)
= a_i$, for $i = 1,2$.
\begin{enumerate}
	\item \textbf{Maximum social welfare.}
$f$ maximizes social welfare: $f_1(G) + f_2(G)  = \MSW(G)$.
	\item \textbf{Shift invariance.} Shifting payoffs by constants
leads to a corresponding shift in the value.  
That is, if $c = (c_1,c_2)\in \IR^2$, 
$G = (\{1,2\},A, u)$ and $G^c = (\{1,2\},A,u^c)$, where $u_i^c(a) =
u_i(a) + c_i$ for 
all $a \in A$, then $f(G^c) = (f_1(G) + c_1, f_2(G) + c_2)$. 
	\item \textbf{Monotonicity in actions.} 
Removing an action of a player cannot increase her value.  
That is, if $G = (\{1,2\}, A_1 \times A_2, u)$, 
and $G' = (\{1,2\}, A_1'\times A_2, u|_{A_1' \times A_2})$, 
where $A_1' \subseteq A_1$, then $f_1(G') \le f_1(G)$, and
similarly if we replace $A_2$ by $A_2' \subseteq A_2$.
	\item \textbf{Payoff dominance.} If, for all action profiles
$a\in A$, a player's expected payoff is strictly larger than her
opponent's, then her value should be at least as large as the
opponent's.  That is, if $u_i(a) \ge u_j(a)$ for all $a \in A$, then
$f_i(G) \ge f_{j}(G)$.
	\item \textbf{Invariance to replicated strategies.} 
Adding a mixed strategy of player 1 as a new action for her does 
not change the value of the game; similarly for player 2. That is, if 
$G = (\{1,2\}, A_1 \times A_2, u)$, 
$t\in S_1$, 
and  $G' = (\{1,2\}, A_1' \times A_2, u')$, where $A_1' = A_1 \cup \{t\}$, 
$u'(t,a_2) = U(t,a_2)$ for all $a_2\in A_2$, and
$u'(a) = u(a)$ for all $a \in A$
(so that $G'$ extends $G$ by adding to $A_1$ one new action, which can
be identified with a mixed strategy in $S_1$). Then $f(G) = f(G')$.  The same
holds if we add a 
strategy to $A_2$.
\end{enumerate}

\begin{thm}\label{cocochar} \textrm{\cite{KK09}}  Axioms 1-5
characterize the coco value in 2-player normal-formal games.%
\footnote{
Kalai and Kalai actually consider Bayesian games in their
characterization, and have an additional axiom that they call
\emph{monotonicity in information}.  This axiom trivializes 
in normal-form games (which can be viewed as the special case of
Bayesian games where players have exactly one possible type).  It is
easy to see that their proof shows that Axioms 1-5 characterizes the
coco value in normal-form games.}
\end{thm}

Note that, following Kalai and Kalai \citeyear{KK09}, we have stated the
axioms for the coco value in terms of the underlying game $G$.  Since,
as we have argued, Kalai and Kalai are assuming there are side payments,
we might consider
stating the axioms in terms of $G^*$.  We could certainly replace all
occurrences of $f_i(G)$ by $f_i(G^*)$; nothing would change if we did
this, since, by Theorem~\ref{coco}, $\coco(G) = \coco(G^*)$.   But we
could go further, replacing $G$, $A$, and $u$ uniformly by $G^*$, $A^*$,
and $u^*$.  
For example, Axiom 1 would say $f_1(G^*) + f_2(G^*) = \MSW(G^*)$;
Axiom 2 would say that $f((G^*)^c) = (f_1(G^*) + c_1, f_2(G^*) + c_2)$.
It is not hard to check that the resulting axioms are
still sound.  Moreover, for all axioms but Axiom 4 (payoff
dominance), the resulting axiom is essentially equivalent to the
original axiom.  (In the case of shift invariance, this is because
$(G^*)^c = (G^c)^*$.)
However, the version of Axiom 4 for $G^*$ is vacuous.
No matter what the payoffs are in $G$, it cannot be the case that a player's
expected payoff is larger than his opponent's for all actions in $G^*$,
since players can always agree to a deal action that results in the
opponent getting a 
large transfer.  Thus, we must express payoff dominance in terms of $G$
in order to prove Theorem~\ref{cocochar}.

We now characterize the M-PCE value axiomatically.  The M-PCE value of
$G$ is not equal to that of $G^*$ in general.  Since we want to compare
the M-PCE value and coco value, it is most appropriate to consider games
with side payments.  Thus, in the axioms for M-PCE, we  write $f_i(G^*)$ rather
$f_i(G)$.  We start by considering the extent to which the M-PCE value
satisfies the axioms above for coco value, with $f_i(G)$ replaced by
$f_i(G^*)$.  As we noted, this change has no impact for coco value
except in the case of Axiom 4 (payoff dominance).  But
Example~\ref{xam1} shows that the M-PCE value does not satisfy payoff
dominance.  The following result shows that it satisfies all the
remaining axioms.

\begin{thm}\label{a1}
The function mapping 2-player games with side payments to their (unique)
M-PCE value 
satisfies maximum social welfare, shift invariance, monotonicity in
actions, and invariance in 
replicated strategies.
\end{thm}
\begin{proof}
We consider each property in turn:
\begin{itemize}
\item The fact that the function satisfies maximum social welfare is
immediate from the characterization in Theorem~\ref{t3}. 
\item  It is easy to see that 
$\MSW(G^c) = \MSW(G) + c_1 + c_2$, $\mm_1(G^c) = \mm_1(G) + c_1$
and $\mm_2(G^c) = \mm_2(G) + c_2$.  It then follows from 
Theorem~\ref{t3} that the M-PCE value of $(G^c)^*$ is the result of
adding $c$ to the M-PCE value of $G^*$.
\item Let $G'$ be as in the description of Axiom 3.
It is almost immediate from the definitions 
that $\MSW(G') \le \MSW(G)$, $\mm_1(G') \le \mm_1(G)$, and $\mm_2(G')
\ge \mm_2(G)$.  The result now follows from Theorem~\ref{t3}.
\item Let $G'$ be the result of adding a replicated action to $S_1$, as
described in the statement of Axiom 5.
Clearly $\MSW(G') = \MSW(G)$, $\mm_1(G') = \mm_1(G)$, and $\mm_2(G') =
\mm_2(G)$.  
Again, the result now follows from Theorem~\ref{t3}.
\end{itemize}
\end{proof}

Our goal now is to axiomatize the M-PCE value in games with side
payments.  Since the M-PCE value and the coco value are different in
general, there must be a difference in their axiomatizations.
Interestingly, we can capture the difference by replacing payoff dominance
by another simple axiom:

\begin{enumerate}
\item[6.] \textbf{Minimax dominance.} If a player's minimax value is no
less than her opponent's minimax value, then her value is no less than
her opponent's. 
That is, if $\mm_i(G) \ge \mm_j(G)$, then $f_i(G^*) \ge f_j(G^*)$. 
\end{enumerate}

It is immediate from Theorem~\ref{t3} that the M-PCE value satisfies
minimax dominance; Example~\ref{xam1} shows that the coco value does not
satisfy it.
We now prove that the M-PCE value is characterized by Axioms 1, 2, and
6.  (Although Axioms 3 and 5 also hold for the M-PCE value, we do not
need them for the axiomatization.)
Interestingly, for all these axioms, we can replace all occurrences of
$G$, $A$, and $u$ by $G^*$, $A^*$, and $u^*$, respectively, to get an equivalent axiom; it really does not
matter if we state the axiom in terms of $G$ or $G^*$ (although the
argument to $f$ must be $G^*$).

\begin{thm}\label{a6}
Axioms 1, 2, and 6 characterize the M-PCE value in 2-player games with
side payments.
\end{thm}
\begin{proof}
Theorem \ref{a1} shows that the M-PCE value satisfies Axioms 1 and 2.
As we observed, the fact that the M-PCE value satisfies Axiom 6 is
immediate from Theorem~\ref{t3}.

To see that the M-PCE value is the unique mapping that satisfies Axioms
1, 2, and 6, suppose that $f$ is a mapping that satisfies these axioms.
We want to show that $f(G^*)$ is the M-PCE value for all games $G$.  
So consider an arbitrary game $G$ such that the M-PCE value of $G^*$ is
$v = (v_1,v_2)$.  By shift invariance, the M-PCE value of $(G^{-v})^*$
is $(0,0)$.  By Axiom 1, $\MSW(G) = v_1 + v_2$, so $\MSW(G^{-v}) = 0$.  
Note that it follows from Theorem~\ref{t3} that $0 = \MSW(G^{-v}) +
\mm_1(G^{-v}) - \mm_2(G^{-v})$.  Since $\MSW(G^{-v}) = 0$, it follows
that $\mm_1(G^{-v}) = \mm_2(G^{-v})$.  Suppose that $f((G^{-v})^*) =
(v_1',v_2')$.  By Axiom 1, we must have $v_1' + v_2' = 0$.  By Axiom 6,
since $\mm_1(G^{-v}) = \mm_2(G^{-v})$, we must have $v_1' = v_2'$.
Thus, $f((G^{-v})^*) = (0,0)$.  By shift invariance, 
$f(G^*) = f((G^{-v})^*)+v=(v_1,v_2)$, as desired. 
\end{proof}

Again, we conclude this subsection by considering what happens if a default
payoff is used instead of backup actions when defining
games with side payments.  
It is still the case that the M-PCE value satisfies Axioms 1, 2, 3, and
5, and does not satisfy Axiom 4.
To get an axiomatization of the M-PCE value in such games with side
payments, we simply need to change Axiom 6 (Minimax Dominance) so that
it uses the default value rather than the minimax value: if the
default value of a player is no 
less than the default value of the opponent, then the player's value is
no less than the opponent's value. Thus, variations in the notion of
games with side payments lead to straightforward variations in the
characterization of the M-PCE value.

\subsection{Complexity comparison}\label{sec:complexityComparison} 
In this section, we consider the complexity of computing the M-PCE value
and the coco value, and the corresponding strategy profiles.

It follows easily from the characterization in Theorem \ref{coco}
that in a 2-player game $G$ with (or without) side payments, 
the coco value 
is determined by $\MSW(G)$, $\mm_1(G_z)$, and $\mm_2(G_z)$.
$G_z$ can clearly be determined from $G$ in polynomial time (polynomial
in the number of strategies), and $\MSW(G)$ can be determined in
polynomial time (simply by inspecting the payoff matrix for $G$).  
The minimax value of a 2-player game can be
computed in polynomial time
(see Appendix~\ref{e}).
Moreover, if $(c_1,c_2)$ is the coco value of $G$,
and $s^*$ is a pure strategy profile that obtains $\MSW(G)$,
the strategy profile that gives players the coco value is
$((s^*,U_1(s^*)-c_1),(s^*,U_1(s^*)-c_1))$, which is simply the deal
strategy profile 
in which both players agree to play $s^*$, and agree that player 1
pays player 2  $(U_1(s^*)-c_1)$.

Similarly, we can compute a M-PCE in a 2-player game with side payments
in polynomial time.
\begin{thm}\label{PCE2}
In a 2-player game $G^*$ with side payments, we can compute its M-PCE value
and a strategy profile that obtains it
in polynomial time.
 \end{thm}
\begin{proof}
Let $G$ be the game underlying $G^*$.  By Theorem \ref{t3}, the M-PCE
value of $G$ 
is entirely determined by its MSW and its minimax value. We show in
Appendix \ref{e} that 
is determined by $\MSW(G)$, $\mm_1(G)$, and $\mm_2(G)$. 
Since the minimax value of a 2-player game can be computed in polynomial
time, and $\MSW(G)$ can be computed by simply finding the entry in the
matrix with the highest total utility,
the M-PCE value can be computed in polynomial time.

Let the M-PCE value be $(m_1,m_2)$, and let $s^*$ be a pure strategy profile
that obtains $\MSW(G)$. Then 
$((s^*,U_1(s^*)-m_1),(s^*,U_1(s^*)-m_1))$, which is simply the deal strategy profile 
in which both players agree to play $s^*$, and agree that player 1 pays player 2
$U_1(s^*)-m_1$, 
is a M-PCE. 
\end{proof}

For 2-player games (without side payments), a PCE can be found in 
polynomial time if one exists; moreover, determining whether one exists
can also be done in polynomial time (see Theorem \ref{PCE}). Similarly, in 
2-player games, a M-PCE can always be found in polynomial time
(see Theorem \ref{MPCE}).

\section{Related Work}\label{sec:related}
There are many  solution concepts in the literature that attempt to
model cooperative play. We compared PCE to the coco value in some detail in
Section~\ref{sec:coco}.  In this section, we compare PCE to a number of
others.

Although PCE is meant to apply to one-shot games,
our motivation for it involved repeated games.  It is thus interesting
to compare Cooperative Equilibrium to solutions of repeated games.
The well-known \emph{Folk Theorem} \cite{OR94} 
says
that any payoff profile that gives each player at least his minimax
utility is the payoff profile of some NE in the repeated game.
Moreover, the proof of the Folk Theorem shows that if $s$ is a strategy
in the underlying normal-form game where each player's utility is higher
than the minimax utility in the repeated game, then there is a NE in the
repeated game where $s$ is played at each round.
Thus, playing cooperatively repeatedly in the repeated game will
typically be an outcome of a NE.
However, so will many other behaviors.  Because so many behaviors are
consistent with the Folk Theorem, it has very little predictive power. 
For example, 
in repeated Traveler's Dilemma, a player can ensure a payoff of at least
2 per iteration simply by always playing 2.  
It follows from the Folk Theorem that
for any strategy profile $s$ in the one-shot game where each player gets
at least 2, there is a NE in the repeated game where each player $i$
plays $s_i$ in each round.
By way of contrast, as we have seen, in a PCE of the single-shot game,
each player gets more than 98.  
More generally, we can show that, for each PCE $s$ in a normal-form game,
there is a NE of the repeated game where $s$ is played repeatedly.

Halpern and Pass \citeyear{HaPa13} and Capraro and Halpern
\citeyear{CH2014} consider what they call \emph{translucent players}, who
believe that how other players respond may depend in part on what they
do.  This is implicitly the case in PCE as well. 
The notion of translucency assumes that 
each player $i$ has beliefs regarding
how other players would respond if $i$ deviates from his intended 
strategy to another strategy. 
That is, for each pair of strategies $(s_i, s_i')$ for player $i$, 
$i$ assigns a probability $\mu_i^{s_i,s_i'}(s_{-i})$ 
to each (joint) strategy profile $s_{-i}$ for players other than $i$. 
Intuitively, $\mu_i^{s_i,s_i'}(s_{-i})$ is the probability at which
player $i$ believes 
the others would jointly play $s_{-i}$, if $i$ deviated from $s_i$ to $s_i'$.
A strategy profile is a \emph{translucent equilibrium (TE)} if there does not exist a player $i$
such that $i$ can strictly improve her payoff if $i$ deviates and other players 
respond to the deviation according to $i$'s belief (of how they would 
respond to the deviation).
In 2-player games,
every PCE is a TE, one in which each player believes that 
the other player would best respond to a deviation;
similarly, every CE  is a TE, one in which each player believes that
the other player  
best responds to a deviation if that makes the other player no worse off 
compared to when no one deviates, and otherwise punishes
the deviation by playing a strategy that makes the one who deviates
strictly worse off than in the case where no one deviates whenever possible.
In $n$-player games for $n > 2$,
every PCE is a TE in which each player believes that 
if she deviates,
the other players would play a NE among themselves given the deviation.
(Recall that CE is defined only for 2-player games.)
However, it is not the case that every TE is a PCE.

\emph{Farsighted pre-equilibrium (FPE)} \cite{JM2011}, like PCE, allows
players to react to what other players are doing.  
Very roughly speaking, while PCE assumes that if a player deviates,
the other players get to best respond, in FPE, the player who deviates
gets to make the final response.
For example, suppose that Alice deviates from $s$ to $s'$. PCE
considers how Bob would react to the deviation, and whether Alice is better
or worse off given Bob's response. 
FPE also considers how Bob would react, 
but allows Alice to take the last step, and then compares Alice's payoff in 
$s$ to her payoff at the end of this process. 
PCE also allows a player $i$ to deviate to a strategy
that may (temporarily) decrease $i$'s payoff
(this could be useful because the response to the deviation may make $i$ 
better off);  FPE does not consider such deviations.
Every NE is an FPE; as we have seen, not every NE is a PCE.
As a consequence, in games like the centipede game, PCE and M-PCE do a
better job of predicting cooperative behavior than FPE.
The concept of farsightedness in FPE, which allows players to consider
other players' responses 
and responses to other players' responses, and so on,  
dates back to von Neumann and Morgenstern's \emph{stable set} in
coalitional games \citeyear{NM44}. 
The idea was then developed by Harsanyi who proposed indirect
dominance of coalition structures \citeyear{Harsanyi74}, 
and then followed by a number of works
\cite{chwe94,DX2003,Greenberg90,Nakanishi2007,SM2005}.  
However, all these works except FPE consider cooperative games instead
of non-cooperative games -- which are the main topic of these paper.

There have also been attempts to explain cooperative behavior by
saying that the utility function that players use is different from
the utility function that is presented in the game, and takes into
account fairness and/or social welfare.  The two best-known examples
of this approach are due to Charness and Rabin \citeyear{CR2002} and
Fehr and Schmidt \citeyear{Fe-Sc}.
Given utility functions $u_i$ for players $i = 1,\ldots,n$, Charness
and Rabin \citeyear{CR2002} consider the modified 
utility functions
\newcommand{\alphaCR}{a^{CR}}
\newcommand{\deltaCR}{b^{CR}}
$$u^{CR}_i =
(1-\alphaCR_i)u_i(s)+\alphaCR_i(\deltaCR_i\min_{j=1,\ldots,N}u_j(s)+(1-\deltaCR_i)\sum_{j=1}^Nu_j(s)),$$
where $\alphaCR_i$ is the degree of importance that agent $i$ gives
to social welfare and the plight of the worst-off individual (so that
$(1-\alpha)$ is the degree of importance of his base utility to player
$i$), while $\deltaCR_i$ measures the relative degree of importance of the 
worst-off individual and $(1-\deltaCR_i)$ measures the relative degree
of importance of total social welfare.
\newcommand{\alphaEF}{a^{FS}}
\newcommand{\betaEF}{b^{FS}}
Similarly, Fehr and Schmidt \citeyear{Fe-Sc} modify the utility to
$$u_i^{FS}(s)=u_i(s)-\frac{\alphaEF_i}{n-1}\sum_{j\neq
  i}\max(u_j(s)-u_i(s),0)-\frac{\betaEF_i}{n-1}\sum_{j\neq
  i}\max(u_i(s)-u_j(s),0),$$ 
where $\betaEF_i\leq\alphaEF_i$, $\alphaEF_i$ can be viewed as
measuring the importance of the inequity caused by $i$ having a lower
payoff than others, and $\betaEF_i$ can be viewed as measuring the
importance of the inequity caused by $i$ having a higher payoff than others.
As shown in Section \ref{sec:coco}, 
M-PCE is closely related to maximal social welfare, and also embodies a
certain sense of fairness, so to some extent it captures some of the
features that the modified utility functions of Charness and Rabin
\citeyear{CR2002} and Fehr and Schmidt \citeyear{Fe-Sc} are trying to capture.

While not intended to model cooperation, the 
recently-introduced notion of \emph{iterated regret minimization} (IRM)
\cite{HP11} often produce results similar to PCE.
As its name suggests, IRM iteratively deletes strategies 
that do not minimize regret. 
Although it based on a quite different philosophy than PCE or its variants,
IRM leads to quite similar predictions as PCE in a surprising number of games.
For example, in Traveler's Dilemma, 
(97, 97) is the unique profile that survives IRM.
In the Nash bargaining game, 
(50, 50) is the unique profile that survives IRM
and is also the unique M-PCE of the game. 
There are a number of other games of interest where PCE and IRM either
coincide or are close.

There are also games in which they behave differently. For example,
consider a  
variant of Prisoner's Dilemma with the following payoff matrix:
\begin{table}[htb]
\begin{center}
\begin{tabular}{c |  c c}
& Cooperate & Defect\\
\hline
Cooperate &(10000,10000) &(0,10001) \\
Defect  &(10001,0)  &(1,1)  
\end{tabular}
\end{center}
\end{table}
\noindent It can be shown that, if there are dominant actions in a game, then
these are the only actions that survive IRM.  Since defecting is the
only dominant action in this game, it follows that (Defect, Defect) is
the only strategy profile that survives IRM, giving a payoff 
(1, 1). On the other hand, 
the unique M-PCE is (Cooperate, Cooperate) with payoffs (10000, 10000)
(although (Defect, Defect) is also a PCE).
In this game, M-PCE seems to do a better job of explaining behavior than
PCE. 
Nevertheless, the fact that PCE and IRM lead to similar answers in so
many games of interest suggests that there may be some deep connection
between them.  We leave the problem of explaining this connection to
future work.

\appendix
\section{Computing the PCE in the centipede game}\label{a}
To compute the PCE in the centipede game, we need to first compute 
$\BU_1$ and $\BU_2$.
If player 1 continues to the end of the
game, then player 2's best response is to also continue to the end of
the game, giving player 1 a payoff of $2^{19}$ (and player 2 a payoff of
$2^{20}+1$).  If we take $q_{i,j}$ to be the strategy where
player $i$ quits at turn $j$ and $q_{i,C}$ to be the strategy where player $i$
continues to the end of the game, then a straightforward computation
shows that $q_{2,C}$ continues to be a best response to $\alpha q_{1,19}
+ (1-\alpha)q_{1,C}$ as long as $\alpha \ge \frac{3\times
2^{18}}{3\times 2^{18}+1}$.  If we take $\alpha =  \frac{3\times
2^{18}}{3\times 2^{18}+1}$ and player 2 best responds by playing
$q_{2,C}$, then player 1's utility is $2^{19}+\frac{3\times
2^{18}}{3\times 2^{18}+1}$.  It is then straightforward to show that
this is in fact $\BU_1$.  A similar argument shows 
\fullv{that, if player 1 is
best responding, then the best player 2 can do is to play $\beta q_{2,18}
+ (1-\beta)q_{2,C}$, where $\beta = \frac{3\times 2^{17}}{3\times 2^{17}+1}$.
With this choice, player 1's best response is $q_{1,19}$;
using this strategy for player 2, we get that}
\shortv{that}
$\BU_2 =
2^{18}+\frac{3\times 2^{17}}{3\times 2^{17}+1}$.

It is easy to see that there is no pure strategy profile $s$ such that
$U_1(s) \ge \BU_1$ and $U_2(s) \ge \BU_2$.  However, there are many
mixed PCE.  
For example,
every strategy profile $(q_{1,C}, s_2)$ where $s_2=\beta
q_{2,18}+(1-\beta)q_{2,C}$ and $\beta \in 
[1-\frac{3\times 2^{17}}{(3\times 2^{17}+1)(3\times 2^{18}+1)}, 
\frac{3\times 2^{18}}{3\times 2^{18}+1}]$ is a PCE.

\section{Proof of Lemma \ref{LEM3}}
\relem{LEM3}
In a 2-player game, for all strategy profiles $s^*$, there exists a
strategy profile $s'=(s_1',s_2')$ that Pareto dominates $s^*$ such that
both $s_1'$ and $s_2'$ have support of size at most two.
\erelem
\begin{proof}
Let $\mathbf{A}$ and $\mathbf{B}$ be the payoff matrices (of size
$n\times m$) for player 1 and player 2 respectively. 
Given a strategy profile $s^* = (s_1^*,s_2^*)$, let $U_1(s^*) =r_1^*$ and
$U_2(s^*)=r_2^*$. 
We first show that there exists a strategy $s_2'$ for player 2 with 
support of size at most two such that 
$(s_1^*,s_2')$ Pareto dominates $s^*$. We then show that
there exists a strategy $s_1'$ for player 1 with support of size at most
two such that $(s_1',s_2')$ Pareto dominates $(s_1^*,s_2')$, and hence
$s^*$.

Consider the following linear program $P_1$, 
where $y$ is a column vector in $\IR^m$:
$$\begin{array}{ll}
\text{maximize} &(s_1^*)\mathbf{A} y \label{e36}\\
\text{subject to}&	(s_1^*)^T\mathbf{B} y = r_2^* \label{e34} \\
& \sum_{i=1}^m y[i]=1 		 \label{e35} \\
&	y \ge 0. \label{e37} 
\end{array}$$
As usual, an \emph{optimal solution} of $P_1$ is a vector $y$ that
maximizes the objective function ($(s_1^*)\mathbf{A} y$)
and satisfies the three constraints;
a \emph{feasible solution} of $P_1$
is one that satisfies the constraints; finally, an
\emph{optimal value} of $P_1$ is the value of the objective function for
the optimal solution $y$ (if it exists).  
We show that $P_1$ has an optimal solution $y^*$ with at most
two nonzero entries.  

Since all constraints in $P_1$ are equality
constraints except for the non-negativity constraint, 
$P_1$ is a \emph{standard-form linear program} \cite{Murty1983}.  
We can rewrite the equality constraints in $P_1$ as
$$
		\mathbf{D}  y = \left[
\begin{array}{l}
r_2^*\\
1
\end{array}\right],
$$
where $\mathbf{D}$ is an $(m\times 2)$ matrix whose first row is 
$(s_1^*)^T \mathbf{B}$ and whose second row has all entries equal to 1. 
In geometric terms, the region represented by 
the constraints in $P_1$ is a convex polytope. Since $P_1$ is a
standard-form linear 
program, it is well-known that $y$ is a vertex of the polytope (i.e., an
extreme point of the polytope) iff all columns $i$ in
$\mathbf{D}$ where $y[i]\neq 0$ are linearly independent \cite{Murty1983}. 
Since the columns of {\bf D} are vectors in $\IR^2$, at most two of them
can be linearly independent.  Thus, a vertex $y$ of the polytope can have
at most two nonzero entries.

Clearly $s_2^*$ is a feasible solution of $P_1$. Since
$(s_1^*)\mathbf{A} s_2^*=r_1^*$, by assumption, the optimal value of $P_1$
is at least $r_1^*$.  Moreover,
since the objective function of $P_1$ is linear, $y\ge 0$, and $\sum_{i=1}^m 
y[i]=1$, the optimal value is bounded. Therefore, the
linear program has an optimal solution.  
By \textit{the fundamental theorem of linear programming}, if a linear
program has an optimal solution,
then it has an optimal solution at a vertex of the polytope defined by its
constraints \cite{Murty1983}.  
Let $s_2'$ be the strategy defined by an optimal solution at the vertex
of the polytope.
As we observed above, $s_2'$ has at most two nonzero entries.
It is immediate that 
$U_1((s_1^*,s_2'))\ge r_1^*$ and
$U_2((s_1^*,s_2'))\ge r_2^*$. 

This completes the first step of the proof. 

The second step of the proof essentially repeats the first step. 
Suppose that $U_1((s_1^*,s_2'))= r_1$ and $U_2((s_1^*,s_2'))=r_2$.
Consider the following linear program $P_2$, where $x$ is  column vector
in $\IR^n$:
$$\begin{array}{ll}
\text{maximize} &x^T\mathbf{B}s_2'\\
\text{subject to } &x^T\mathbf{A}s_2'= r_1\\
&\sum_{i=1}^n x[i]=1\\
&	x \ge 0. 
\end{array}$$
Since $s_1^*$ is a feasible solution of $P_2$ and $(s_1^*)^T\mathbf{B}
s_2'\ge r_2^*$, the optimal value of $P_2$ is at least $r_2^*$.
As above, if we take $s_2'$ to be an optimal solution of $P_2$ that is a
vertex of the polytope
defined by the constraints, then $s_2'$ has support of size at
most two, and $(s_1',s_2')$ Pareto dominates $s^*$.
\end{proof}

\section{Proof of Lemma \ref{BP}}
\relem{BP}
A simple bilinear program of size $2 \times 2$ can be solved in constant
time.  
\erelem
\begin{proof}
Let $P$ be the following simple bilinear program, where $x = [x_1 \ 
x_2]^T$, $y = [y_1 \ y_2]^T$: 
$$\begin{array}{ll}
\text{maximize} &x^T\mathbf{A}y+ x^T  c + y^T c'\\
\text{subject to} &x^T \mathbf{B} y \ge d_1\\ 

                  & x_1 + x_2 = d_2\\ 
                  &y_1 + y_2  = d_3\\
                  & x \ge 0\\
                  &y \ge 0,
\end{array}$$
where $\mathbf{A}$ and $\mathbf{B}$ are $2 \times 2$ matrices.  

We show that $P$ can be solved in constant time. That is, we either find
an optimal solution of $P$, or find that $P$ has no optimal solution in
constant time. The idea is to show that $P$ can be reduced into eight
simpler problems, each of which can more obviously be solved in constant time.

Suppose that
$\mathbf{A}= \left[ \begin{array}{cc}
a_{11} & a_{12}\\
a_{21} & a_{22}\end{array}\right]$ and
$\mathbf{B}= \left[ \begin{array}{cc}
b_{11} & b_{12}\\
b_{21} & b_{22}\end{array}\right].$
Then we can write $P$ as the following quadratic program $Q$: 
\begin{equation}
\begin{aligned}
\text{maximize}\ &a_{11}x_1y_1 + a_{12}x_1y_2 + a_{21}x_2y_1 +
a_{22}x_2y_2+ c[1]x_1+ c[2]x_2+c'[1]y_1 + c'[2]y_2\\
\text{subject to }
& b_{11}x_1y_1 + b_{12}x_1y_2 + b_{21}x_2y_1 + b_{22}x_2y_2 -d_1\ge 0\\
& x_1+x_2 = d_2\\
& y_1+y_2= d_3\\
&x_1,x_2,y_1,y_2 \ge 0.
\end{aligned}\nonumber
\end{equation}
After replacing $x_2$ with $(d_2-x_1)$ and $y_2$ with $(d_3-y_1)$, then
rearranging terms, the objective function of $Q$ becomes
$$\begin{aligned}
&(a_{11}-a_{12}-a_{21}+a_{22})x_1y_1
+(a_{12}d_3-a_{22}d_3+c[1]-c[2])x_1+\\
&(a_{21}d_2-a_{22}d_2+c'[1]-c'[2])y_1 + (a_{22}d_2d_3+c[2]d_2+c'[2]d_3),
\end{aligned}$$
and the first constraint becomes
$$\begin{aligned}
(b_{11}-b_{12}-b_{21}+b_{22})x_1y_1+ (b_{12}d_3-b_{22}d_3)x_1+
(b_{21}d_2-b_{22}d_2)y_1+(b_{22}d_2d_3-d_1).
\end{aligned}$$

We can get an equivalent problem by removing 
the constant terms $a_{22}d_2d_3+c[2]d_2+c'[2]d_3$ from the objective
function, since 
adding or removing   
additive constants from a function that we want to maximize does not
affect its optimal solutions (e.g., ``$\text{maximize}\  x$'' has the same optimal solutions as ``$\text{maximize}\  (x+1)$'').

Thus, $Q$ is equivalent to the following quadratic program $Q'$:
$$
\begin{aligned}
\text{maximize}\ \ &\gamma_1x_1y_1+\gamma_2x_1+\gamma_3y_1\\
\text{subject to }
& \gamma_4 x_1y_1 + \gamma_5 x_1 + \gamma_6 y_1 +\gamma_7\ge 0\\
&x_1\in [0,d_2]\\
&y_1\in [0,d_3],\\
\end{aligned}
$$
where 
\begin{equation}
\begin{array}{l}
\gamma_1= a_{11}-a_{12}-a_{21}+a_{22}\\
\gamma_2= a_{12}d_3-a_{22}d_3+c[1]-c[2]\\
\gamma_3= a_{21}d_2-a_{22}d_2+c'[1]-c'[2]\\
\gamma_4= b_{11}-b_{12}-b_{21}+b_{22}\\
\gamma_5= b_{12}d_3-b_{22}d_3\\
\gamma_6= b_{21}d_2-b_{22}d_2\\
\gamma_7= b_{22}d_2d_3-d_1.
\end{array}\nonumber
\end{equation}
(Note that $\gamma_i$ is a constant, for $i = 1,\ldots,7$.)

The first step in solving $Q'$ involves expressing the values of $y_1$
that make $(x_1,y_1)$ a feasible solution, that is, one that satisfies
the constraint 
$$\gamma_4 x_1y_1 + \gamma_5 x_1 + \gamma_6 y_1 +\gamma_7 = (\gamma_4 y_1
+\gamma_5)x_1 + \gamma_5 x_1 + \gamma_6 y_1 +\gamma  \ge 0.$$
For each $y_1\in [0,d_3]$, let $\Psi_1(y_1)$ be the set of $x_1$
such that $(x_1,y_1)$ is a feasible solution of $Q'$.  The characterization
of $\Psi_1(y_1)$ depends on the sign of $\gamma_4 y_1 +\gamma_5$.
Specifically:    
\begin{equation}\label{cases}
	\left\{
	\begin{array} {ll}
			\Psi_1(y_1)=
              \left[\frac{-\gamma_6y_1-\gamma_7}{\gamma_4y_1+\gamma_5},d_2\right]
\cap[0,d_2]
&\text{ if } \gamma_4y_1+\gamma_5>0,
		\frac{-\gamma_6y_1-\gamma_7}{\gamma_4y_1+\gamma_5}\le d_2, \\
		 	\Psi_1(y_1)=
                        \left[0,\frac{-\gamma_6y_1-\gamma_7}{\gamma_4y_1+\gamma_5
}\right]\cap[0,d_2]
&\text{ if } \gamma_4y_1+\gamma_5<0, 
		\frac{-\gamma_6y_1-\gamma_7}{\gamma_4y_1+\gamma_5}\ge 0, \\
			\Psi_1(y_1)= \left[0,d_2\right] &\text{ if }
                        \gamma_4y_1+\gamma_5=0,
Q
                        \gamma_6y_1+\gamma_7\ge 0,\\ 
			\Psi_1(y_1)= \emptyset, &\text{ if }
	\gamma_4y_1+\gamma_5=0, \gamma_6y_1+\gamma_7< 0.  
	\end{array}\right. 
\end{equation}
Note that the first three regions are single intervals.

Let $f(x_1,y_1)=\gamma_1x_1y_1+\gamma_2x_1+\gamma_3y_1$, so that $f(x_1,y_1)$ is the objective
function of $Q'$.  
We want to maximize $f$ over all feasible pairs $(x_1,y_1)$.  
Taking the derivative of $f$ with respect to $x_1$, we get
\begin{equation}
\frac{\partial f(x_1,y_1)}{\partial x_1} = \gamma_1y_1+\gamma_2,\nonumber
\end{equation}
which is a linear function of $y_1$.  
Because the derivative is linear, for each fixed value of $y_1$, 
the value that maximizes $f(x_1,y_1)$ must lie at an endpoint of the
interval appropriate for that value of $y_1$.  Whether it is the left
endpoint or the right endpoint depends on whether the derivative is
negative or positive.
For example, if $y_1$ satisfies the constraints corresponding to the
first interval in (\ref{cases}) (i.e., if 
$\gamma_4y_1+\gamma_5>0$ and
$\frac{-\gamma_6y_1-\gamma_7}{\gamma_4y_1+\gamma_5}\le d_2$)
and $\gamma_1 y_1 + \gamma_2 > 0$, 
then   
$x_1=d_2$ (i.e., the right endpoint of the interval of $\Psi_1(y_1)$)
maximizes $f(x_1,y_1)$;
and the problem of maximizing $f(x_1,y_1)$ reduces to that of maximizing
$f(d_2,y_1)$ (see $Q_1$ below). 
On the other hand, if $\gamma_1 y_1 + \gamma_2 > 0$, then 
maximizing $f(x_1,y_1)$ reduces to maximizing $f(0,y_1)$ or 
$f(\frac{-\gamma_6-\gamma_7}{\gamma_4y_1 + \gamma_5},y_1)$, depending on
whether $\frac{-\gamma_6 y_1 - \gamma_7}{\gamma_4y_1 + \gamma_5}$ is
negative (see $Q_5$ and $Q_6$ below).

These considerations show that to find the value $(x_1,y_1)$ that
maximizes $f(x_1,y_1)$, it suffices to find the value of $y_1$ that
maximizes each of the expressions below, and take the one that is best
among them:
$$
  \begin{array} {ll}
Q_1:   \text{maximize}\  f(d_2,y_1), \text{ subject to }\\\ \ \ \ 
    \gamma_4y_1+\gamma_5>0, \frac{-\gamma_6y_1-\gamma_7}{\gamma_4y_1+\gamma_5}\le d_2, 
\gamma_1 y_1 + \gamma_2  \ge 0, y_1\in [0,d_3]\\
Q_2:   \text{maximize}\  f(0,y_1), \text{ subject to }\\\ \ \ \ 
    \gamma_4y_1+\gamma_5>0, \frac{-\gamma_6y_1-\gamma_7}{\gamma_4y_1+\gamma_5}\le 0, 
    \gamma_1 y_1 + \gamma_2< 0, y_1\in [0,d_3]\\
Q_3:   \text{maximize}\  f(\frac{-\gamma_6y_1-\gamma_7}{\gamma_4y_1+\gamma_5},y_1), \text{ subject to }\\\ \ \ \ 
    \gamma_4y_1+\gamma_5>0, 0\le \frac{-\gamma_6y_1-\gamma_7}{\gamma_4y_1+\gamma_5}\le d_2, 
    \gamma_1 y_1 + \gamma_2< 0, y_1\in [0,d_3]\\
Q_4:   \text{maximize}\  f(d_2,y_1), \text{ subject to }\\\ \ \ \ 
    \gamma_4y_1+\gamma_5<0, \frac{-\gamma_6y_1-\gamma_7}{\gamma_4y_1+\gamma_5}\ge d_2, 
    \gamma_1 y_1 + \gamma_2\ge 0, y_1\in [0,d_3]\\
Q_5:   \text{maximize}\  f(\frac{-\gamma_6y_1-\gamma_7}{\gamma_4y_1+\gamma_5},y_1), \text{ subject to }\\\ \ \ \ 
    \gamma_4y_1+\gamma_5<0, 0\le \frac{-\gamma_6y_1-\gamma_7}{\gamma_4y_1+\gamma_5}\le d_2, 
    \gamma_1 y_1 + \gamma_2\ge 0, y_1\in [0,d_3]\\
Q_6:   \text{maximize}\  f(0,y_1), \text{ subject to }\\\ \ \ \ 
    \gamma_4y_1+\gamma_5<0, \frac{-\gamma_6y_1-\gamma_7}{\gamma_4y_1+\gamma_5}\ge 0, 
    \gamma_1 y_1 + \gamma_2< 0, y_1\in [0,d_3]\\
Q_7:   \text{maximize}\  f(d_2,y_1), \text{ subject to }\\\ \ \ \ 
    \gamma_4y_1+\gamma_5=0, \gamma_6y_1+\gamma_7\ge 0, 
    \gamma_1 y_1 + \gamma_2\ge 0, y_1\in [0,d_3]\\
Q_8:   \text{maximize}\  f(0,y_1), \text{ subject to }\\\ \ \ \ 
    \gamma_4y_1+\gamma_5=0, \gamma_6y_1+\gamma_7\ge 0, 
    \gamma_1 y_1 + \gamma_2< 0, y_1\in [0,d_3].
	\end{array}
$$
Note that $Q_1$, $Q_2$, and $Q_3$ describe the possibilities for the
first case in (\ref{cases}), $Q_4$, $Q_5$, and $Q_6$ are the
possibilities for the second case, and $Q_7$ and $Q_8$ are the
possibilities for the third case.

Each of $Q_1$, $Q_2$, $Q_4$, $Q_6$, $Q_7$, and $Q_8$ can be easily
rewritten as linear 
programs of a single variable ($y_1$), so can be solved in constant time.
With a little more effort, we can show $Q_3$ and $Q_5$ can also be
solved in constant time.  
We explain how this can be done for $Q_3$.  The argument for $Q_5$ is
similar and left to the reader.
All the constraints in $Q_3$ can be viewed as linear
constraints; the set of feasible values of $y_1$ is thus an interval,
whose endpoints can clearly be computed in constant time.
Now the objective function is 
$$f(\frac{-\gamma_6y_1-\gamma_7}{\gamma_4y_1+\gamma_5},y_1) = 
\frac{(\gamma_1y_1+\gamma_2)(-\gamma_6y_1-\gamma_7)}{\gamma_4y_1+\gamma_5} +
\gamma_3y_1.$$
To find the maximum value of the objective function among the feasible
values, we need to take its derivative (with respect to $y_1$).  
A straightforward calculation shows that this derivative is
$$\frac{(-2\gamma_1\gamma_6y_1-\gamma_1\gamma_7-\gamma_2\gamma_6)(\gamma_4y_1+\gamma_5)-
\gamma_4(\gamma_6y_1+\gamma_7)(\gamma_1y_1+\gamma_2)}  
       {(\gamma_4y_1+\gamma_5)^2}+\gamma_3.$$
This derivative is 0 when its numerator is 0 (since the constraints in $Q_3$
guarantee that the denominator is positive).
The numerator is a quadratic, 
so can be solved  in constant time. 

Thus, to find the optimal value for $Q_3$, we must just check $f$ at the
endpoints of the interval defined by the constraints (which, as we
observed above, can be computed in constant time) and at the points where
the derivative is 0 (which can also be computed in constant time).  Thus,
 $Q_3$ can be solved in constant time.  

This completes the argument that $Q$ can be solved in constant time.
\end{proof}

\section{Proof of Theorem \ref{PCE}}
\rethm{PCE}
Given a 2-player game $G = (\{1,2\}, A, u)$, we can compute in polynomial time
whether $G$ has a PCE and, if so, we can compute a PCE in 
polynomial time.
\erethm
\begin{proof}
Suppose that $G = (\{1,2\}, A, u)$, where  $A = A_1 \times A_2$, $|A_1|= n$, 
$|A_2| = m$, $u_1$ is characterized by the payoff matrix $\mathbf{A}$,
and $u_2$ is characterized by the payoff matrix $\mathbf{B}$.

In order to compute a PCE for the game, we need the values of $\BU_1$ and
$\BU_2$.  These can be computed in polynomial time, as follows.
For $\BU_i$, for each $i \in \{1, \ldots, m\}$, we solve the following
linear program $P_i$:
$$\begin{array}{ll}
\text{maximize} &s_1^T(\mathbf{A}[\cdot,i])\\
\text{subject to} &s_1^T(\mathbf{B}[\cdot,i])\ge
	s_1^T(\mathbf{B}[\cdot,j]) \ \ \text{for all}\ j\in \{1,\ldots m\}\\ 
&\sum_{l=1}^{n} s_1[l]=1\\
&s_1\ge 0.\nonumber
\end{array}$$
Suppose that $r_i$ is the optimal value of $P_i$.  Since $P_i$ is a
linear program, $r_i$ can be computed in polynomial time.
Intuitively, $r_i$ is the maximum reward player 1 can get if
player 2 plays action $b_i$ and $b_i$ is a best response for player 2
to 1's action.
(The first constraint ensures that $b_i$ is a best
response for player 2 to player 1's strategy.)
$\BU_1 = \max_{i=1}^m r_i$, so can be computed in polynomial time.
$\BU_2$ can be similarly computed.

After computing $\BU_1$ and $\BU_2$, we can compute a PCE. Recall that a
strategy profile $s$ is a PCE iff $U_1(s)\ge \BU_1$ and $U_2(s)\ge \BU_2$.
Suppose that game $G$ has a PCE $s^*$.  By Lemma \ref{LEM3}, there must exist
a strategy profile $s'=(s_1',s_2')$ that Pareto dominates $s^*$, where
both $s_1'$ and $s_2'$ have support of size at most two.  By
Theorem~\ref{prop1a}, $s'$ is also a PCE.  
We call such a PCE a \emph{$(2\times 2)$-PCE}. 
Our arguments above show that $G$ has a PCE iff it has a $(2\times
2)$-PCE. 
Thus, in order to check whether $G$ has a PCE, it suffice to check
whether it has a $(2\times 2)$-PCE. 

We do this exhaustively.  For all $i_1,i_2\in\{1,2,\ldots,n\}$  with
$i_1 \ne i_2$
and all $j_1,j_2\in\{1,2,\ldots,m\}$ with $j_1\neq j_2$,
we check whether $G$ has a 
$(2\times 2)$-PCE in which player 1 places positive probability only on
strategies $i_1$ and $i_2$, and player 2 places positive probability only on
strategies $j_1$ and $j_2$.  For each choice of $i_1, i_2, j_1, j_2$, this
question can be expressed as the following $2 \times 2$ simple bilinear
programming problem $P_{i_1,i_2,j_1,j_2}$, where $\mathbf{
A}_{i_1,i_2,j_1,j_2}$ is the $2 \times 2$ matrix $ 	\left[
		\begin{array}{cc}
			\mathbf{A}[i_1,j_1] &  \mathbf{A}[i_1,j_2]\\
			\mathbf{A}[i_2,j_1] &  \mathbf{A}[i_2,j_2]
		\end{array} \right]$,
	and $\mathbf{B}_{i_1,i_2,j_1,j_2}$ is the $2 \times 2$ matrix $ 	\left[
		\begin{array}{cc}
			\mathbf{B}[i_1,j_1] &  \mathbf{B}[i_1,j_2]\\
			\mathbf{B}[i_2,j_1] &  \mathbf{B}[i_2,j_2]
		\end{array} \right]$:		

$$\begin{array}{ll}
\text{maximize} &[x_1\  x_2]~\mathbf{A}_{i_1,i_2,j_1,j_2}~
[y_1\  y_2]^T\\
\text{subject to} 	
	&[x_1\  x_2]~\mathbf{B}_{i_1,i_2,j_1,j_2}~
	[y_1\  y_2]^T\ge \BU_2\\
	& x_1+x_2=1\\
	& y_1+y_2=1\\
	& x\ge 0, \ y \ge 0.
\end{array}$$
The first constraint ensures that player 2's reward is at least $\BU_2$;
the remaining constraints ensure that player 1 puts positive probability
only on strategies $i_1$ and $i_2$, while player 2 puts positive
probability only on $j_1$ and $j_2$.
If the optimal value of $P_{i_1,i_2,j_1,j_2}$ for some choice of
of $(i_1,i_2,j_1,j_2)$ is at least $\BU_1$, then the corresponding
optimal solution $(x,y)$ is a PCE of $G$. 
(Recall that a strategy profile $s$ is a PCE if $U_1(s)\ge \BU_1$, and
$U_2(s)\ge \BU_2$.) 
On the other hand, if the optimal value for each $P_{i_1,i_2,j_1,j_2}$
is strictly less than  
$\BU_1$, then $G$ does not have a $(2\times 2)$-PCE and so, by the
arguments above, $G$ does not have a PCE.  

The algorithm above must solve $(^n_2) \times (^m_2)$ simple 2 \time 2
bilinear programs.  By Lemma~\ref{BP}, each can be solved in constant
time.   Thus, the algorithm runs in polynomial time, as desired.
\end{proof}

\section{Proof of Theorem \ref{MPCE}}
\rethm{MPCE}
Given a 2-player game $G = (\{1,2\}, A, u)$, we can compute 
a M-PCE in polynomial time. 
\erethm
\begin{proof}
We start by computing $\BU_1$ and $\BU_2$, as in Theorem \ref{PCE}. Again,
this takes polynomial time. 

Recall that a M-PCE is an $\alpha$-PCE such that for all
$\alpha'>\alpha$, there is no $\alpha'$-PCE in $G$.  
Clearly, a strategy that Pareto dominates an $\alpha$-PCE must itself
be an $\alpha$-PCE.  Thus, using Lemma~\ref{LEM3}, it easily follows
that there must be a M-PCE for $G$ such that the support of both
strategies involved is of size at most 2.  Call such a M-PCE a 
$(2\times 2)$-M-PCE.  

To compute a $(2 \times 2)$-M-PCE, for each tuple $(i_1, i_2, j_1, j_2)$, we
compute the optimal $\alpha$ for which there exists an $\alpha$-PCE when
player 1 is restricted to putting positive probability on actions $i_1$
and $i_2$, and player 2 is restricted to putting positive probability
in $j_1$ and $j_2$.  
Using the notation of Theorem~\ref{PCE}, we want to solve the following
problem $Q_{i_1,i_2,j_1,j_2}$, 
where $d_1(x_1, x_2, y_1, y_2) = [x_1\  x_2]~
\mathbf{A}_{i_1,i_2,j_1,j_2} ~[y_1\  y_2]^T - \BU_1$ and 
$d_2(x_1, x_2, y_1, y_2) = [x_1\  x_2]~
\mathbf{B}_{i_1,i_2,j_1,j_2} ~[y_1\  y_2]^T - \BU_2$:

$$
\begin{array}{ll}
\text{maximize} & \min (d_1(x_1,x_2,y_1,y_2), d_2(x_1,x_2,y_1,y_2))\\
\text{subject to }
 	& x_1+x_2=1\\
	& y_1+y_2=1\\
	& x\ge 0,\ y \ge 0.
\end{array}$$
The objective function maximizes the $\alpha$ for which the strategy
profile determined by $[x_{i_1},x_{i_2}]$ and $[y_{i_1},y_{i_2}]$ is an
$\alpha$-PCE 
(recall that $s$ is an $\alpha$-PCE if
$\alpha=\min(U_1(s)-\BU_1,U_2(s)-\BU_2)$).  
The problem here is that since the objective function involves a $\min$,
this is not a bilinear program.  However, we can solve this problem by
solving two simple bilinear programs of size $2\times 2$, depending on
which of  
$[x_{i_1} x_{i_2}] \mathbf{A}_{i_1,i_2,j_1,j_2}$ $[y_{i_1} y_{i_2}]^T - \BU_1$ 
and 
$[x_{i_1} x_{i_2}] \mathbf{A}_{i_1,i_2,j_1,j_2} [y_{i_1} y_{i_2}]^T -
\BU_2$ is smaller.   
Specifically, let $Q_{i_1,i_2,j_1,j_2}'$ be the following simple bilinear program:
$$\begin{array}{ll}
\text{maximize} &d_1(x_1, x_2, y_1, y_2)\\
\text{subject to }& d_1(x_1, x_2, y_1, y_2) \le
d_2(x_1, x_2, y_1, y_2)\\
 	& x_1+x_2=1\\
	& y_1+y_2=1\\
	& x\ge 0, \ y \ge 0.
\end{array}$$
Let $Q_{i_1,i_2,j_1,j_2}''$ be the same bilinear program with the roles of
$d_1$ and $d_2$ reversed.  It is easy to see that the larger of the
solutions to  $Q_{i_1,i_2,j_1,j_2}'$ and $Q_{i_1,i_2,j_1,j_2}''$ is the
solution to $Q_{i_1,i_2,j_1,j_2}$.  It thus follows that a M-PCE can be
computed in polynomial time.
\end{proof}

\section{Proof of Theorem \ref{POMPCE}}
\rethm{POMPCE}
Given a 2-player game $G = (\{1,2\}, A, u)$, we can compute 
a Pareto-optimal M-PCE in polynomial time. 
\erethm
\begin{proof}
We start by computing a M-PCE $s$, as in Theorem \ref{MPCE}. 
This takes polynomial time.
We then compute a Pareto-optimal strategy profile $s^*$ that Pareto 
dominates $s$. Clearly, $s^*$ is a Pareto-optimal M-PCE, and we are done.

We now show that such an $s^*$ can be found in polynomial time.
We first show that it is impossible to have both 
$U_1(s^*)>U_1(s)$ and $U_2(s^*)>U_2(s)$.
To see why, let $\alpha_s$ be the greatest $\alpha$ such that $s$ is an $\alpha$-PCE.
If $U_1(s^*)>U_1(s)$ and $U_2(s^*)>U_2(s)$,
then $s^*$ is an $\alpha'$-PCE for some $\alpha'$ such that $\alpha'
> \alpha_s$,  
a contradiction to $s$ being a M-PCE.
Therefore, for $s^*$ to Pareto dominate $s$, 
either $U_1(s^*)=U_1(s)$ and $U_2(s^*)\ge U_2(s)$, or $U_1(s^*)\ge
U_1(s)$ and $U_2(s^*)=U_2(s)$. 
It then follows that to find $s^*$, we just need to solve the following 
two bilinear programs $Q_1$ and $Q_2$;
the solution which Pareto dominates the other solution is then Pareto 
optimal (if neither Pareto dominates the other, then both are
Pareto-optimal).
Intuitively, $Q_1$ finds a strategy 
profile that maximizes player 1's reward while player
2 gets no less than what she gets in $s$; and $Q_2$ finds one
that maximizes player 2's reward while player 1 gets no
less than what he gets in $s$.

$Q_1$ is the following bilinear program:
$$\begin{array}{ll}
\text{maximize} &s_1^T\mathbf{A}s_2\\
\text{subject to} &s_1^T\mathbf{B}s_2\ge U_2(s)\\
&\sum_{l=1}^{n} s_1[l]=1\\
&\sum_{l=1}^{m} s_2[l]=1\\
&s_1,s_2\ge 0.\nonumber
\end{array}$$
$Q_2$ is defined similarly, but interchanging $\mathbf{A}$ and $\mathbf{B}$, 
and replacing $U_2$ by $U_1$.

We can use techniques similar to those used in Theorem \ref{PCE} to reduce
both $Q_1$ and $Q_2$ to a polynomial number of simple bilinear
programs. By Lemma \ref{BP}, each simple bilinear program can be solved
in constant time; thus 
both $Q_1$ and $Q_2$ can be solved in polynomial time, as desired. 
\end{proof}

\section{Minimax Value in 2-player games}\label{e}
\begin{thm}\label{mm}
Given a 2-player game $G = (\{1,2\}, A, u)$, we can compute 
$\mm_1(G)$ and $\mm_2(G)$ in polynomial time.
\end{thm}
\begin{proof}
Suppose that $G = (\{1,2\}, A, u)$, where  $A = A_1 \times A_2$, $|A_1|
= n$, 
$|A_2| = m$, $u_1$ is characterized by the payoff matrix $\mathbf{A}$,
and $u_2$ is characterized by the payoff matrix $\mathbf{B}$.

To compute $\mm_1(G)$, 
for each $i \in \{1, \ldots, n\}$, we solve the following
linear program $P_i$:
$$\begin{array}{ll}
\text{minimize} &\mathbf{A}[i,\cdot]\, s_2\\
\text{subject to} &\mathbf{A}[i,\cdot]\, s_2\ge
	\mathbf{A}[j,\cdot]\, s_2 \ \ \text{for all}\ j\in \{1,\ldots n\}\\ 
&\sum_{l=1}^{m} s_2[l]=1\\
&s_2\ge 0.\nonumber
\end{array}$$
Suppose that $r_i$ is the optimal value of $P_i$ (if $P_i$ has a
feasible solution).
Since $P_i$ is a
linear program, $r_i$ can be computed in polynomial time.
Intuitively, $r_i$ is the minimum reward player 1 gets when action $a_i$
is a best response to player 2's strategy.
(The first constraint ensures that $a_i$ is a best
response for player 1 to player 2's strategy.)

It follows that $mm_1(G)=\min_{i=1}^n r_i$, and can be computed in
polynomial time;  $mm_2(G)$ can be computed similarly.
\end{proof}

{\bf Acknowledgements:}
Work supported in part by NSF grants IIS-0534064, IIS-0812045,
IIS-0911036, and CCF-1214844,
by AFOSR grants 
FA9550-08-1-0438,  FA9550-09-1-0266, and FA9550-12-1-0040,
by ARO grants W911NF-09-1-0281 and W911NF-14-1-0017, and 
by the Multidisciplinary
University Research Initiative (MURI) program administered by the AFOSR
under grant FA9550-12-1-0040.

\bibliographystyle{chicagor}
\bibliography{joe,z,nan}  

\end{document}